\documentclass[twoside,journey]{IEEEtran}
\usepackage{makecell}
\usepackage{hyperref}
\usepackage{array}
\usepackage{graphicx,amssymb,amsmath}
\usepackage{multicol}
\usepackage[noadjust]{cite}
\usepackage{setspace}
\usepackage{graphicx}
\usepackage{float}
\usepackage {url}
\usepackage{stfloats}
\usepackage{amsthm,pifont}
\usepackage{flushend}
\usepackage{subcaption}
\usepackage{cases,subeqnarray}
\usepackage{bm,multirow,bigstrut}
\usepackage{amsmath, amsthm, amssymb}
\usepackage{textcomp}
\usepackage{latexsym,bm}
\usepackage{booktabs}
\usepackage{xcolor}
\usepackage{mathtools}
\usepackage{dsfont}
\usepackage{extarrows}
\usepackage{epsfig}
\usepackage{epsfig}
\usepackage{epstopdf}

\usepackage{makecell}
\usepackage{colortbl}
\usepackage{booktabs}
\usepackage{graphicx} % For scaling table to text width
\usepackage{array}    % For centering cells in a column
\usepackage{hyperref} % For hyperlinks
\usepackage{booktabs}
\usepackage{enumitem}
\usepackage{diagbox}
\usepackage{tikz}
\usetikzlibrary{trees,shadows.blur}
\usepackage{forest}
\usepackage{flushend}
\DeclareMathOperator{\Mean}{mean}
\usepackage[ruled]{algorithm2e}
\usepackage{algorithmic} %format of the algorithm
\theoremstyle{plain}

\theoremstyle{plain}

\newtheorem{prop}{Proposition}
\usepackage{amsmath}

\newcommand{\new}[1]{{#1}}

\newcommand{\ignore}[1]{{{\color{yellow} }}}

\renewcommand{\arraystretch}{1.5} % 调整行高

\definecolor{blue-green}{rgb}{0.0, 0.87, 0.87}
\IEEEoverridecommandlockouts
\begin{document}

%----------------------------title&author&thanks----------------------------
% \title{When Networks Meets Generative AI: Potentials, Challenges, and Directions}
% \title{Sensor Placement in Electrical Grid for Anomaly Detection: A Reinforcement Learning Method based on Graph Diffusion Model}
% \title{Wireless Sensor Placement in Cyber‑Physical Power Systems
% for Anomaly Detection via\\
% Dense Reward Graph Diffusion}
\title{Generative AI Enabled Robust Sensor Placement in Cyber-Physical Power Systems: A Graph Diffusion Approach}

\author{Changyuan Zhao, Guangyuan Liu, Bin Xiang,
Benoit Delinchant, Dong In Kim, \textit{Life~Fellow, IEEE}
\thanks{C. Zhao is with the College of Computing and Data Science, Nanyang Technological University, Singapore, and CNRS@CREATE, 1 Create Way, 08-01 Create Tower, Singapore 138602 (e-mail: zhao0441@e.ntu.edu.sg).
}
\thanks{G. Liu is with the College of Computing and Data Science, Nanyang Technological University, Singapore (e-mail: liug0022@e.ntu.edu.sg; dniyato@ntu.edu.sg).}
\thanks{B. Xiang is with Suzhou Institute for Advanced Research, University of Science and Technology of China (e-mail: bin.xiang@ustc.edu.cn)}
\thanks{B. Delinchant is with Univ. Grenoble Alpes, CNRS, Grenoble INP, G2Elab, 38000 Grenoble, France (e-mail: benoit.delinchant@grenoble-inp.fr)}
\thanks{D. I. Kim is with the Department of Electrical and Computer Engineering, Sungkyunkwan University, Suwon 16419, South Korea (e-mail: dongin@skku.edu).}
}

\maketitle
%----------------------------abstract----------------------------
\vspace{-1cm}

\begin{abstract}

With advancements in physical power systems and network technologies, integrated Cyber-Physical Power Systems (CPPS) have significantly enhanced system monitoring and control efficiency and reliability. This integration, however, introduces complex challenges in designing coherent CPPS, particularly as few studies concurrently address the deployment of physical layers and communication connections in the cyber layer. This paper addresses these challenges by proposing a framework for robust sensor placement to optimize anomaly detection in the physical layer and enhance communication resilience in the cyber layer. We model the CPPS as an interdependent network via a graph, allowing for simultaneous consideration of both layers. Then, we adopt the Log-normal Shadowing Path Loss (LNSPL) model to ensure reliable data transmission. Additionally, we leverage the Fiedler value to measure graph resilience against line failures and three anomaly detectors to fortify system safety. However, the optimization problem is NP-hard. Therefore, we introduce the Experience Feedback Graph Diffusion (EFGD) algorithm, which utilizes a diffusion process to generate optimal sensor placement strategies. This algorithm incorporates cross-entropy gradient and experience feedback mechanisms to expedite convergence and generate higher reward strategies. Extensive simulations demonstrate that the EFGD algorithm enhances model convergence by 18.9\% over existing graph diffusion methods and improves average reward by 22.90\% compared to Denoising Diffusion Policy Optimization (DDPO) and 19.57\% compared to Graph Diffusion Policy Optimization (GDPO), thereby significantly bolstering the robustness and reliability of CPPS operations.

\end{abstract}
%----------------------------keywords----------------------------
\begin{IEEEkeywords}
Generative AI, cyber-physical power system, sensor placement, diffusion model.
\end{IEEEkeywords}
%\newpage
\IEEEpeerreviewmaketitle
%----------------------------introduction----------------------------

\section{Introduction}\label{intro}

% \todo{\cite{sun2026generativeintentpredictionagentic}}

% \todo{
% 2. (Link failure due to path loss)
% I am not sure if the link assumption of same signal and noise power holds in practice, as a sensor can be activated when the sensor (received) power 
% is typically greater than 0dBm, at least -20dBm because of circuit turn-on voltage. 
% Given the noise power is in the range of -70dBm and -110dBm, the path loss may be between 60dB and 100dB when the signal power is 0dBm. 
% % For example, if the reference path loss (do = 1m) is 20/30dB, and gamma=2, the path loss is 40/50dB at d=10m and increases to 60/70dB at d=100m.
% 3. (Link failure due to combined path loss and shadowing)
% In equation (3), the combined path loss and shadowing is a random variable, so a link failure needs to be measured as the probability that 
% the combined loss is greater than a given threshold ($lambda_c$). Otherwise, the link activation seems to be varying over time/space. 
% It is not clear how to generate the specific values (realizations) of the combined loss over time/space. }

In previous years, advances in control theory and infrastructure improvements within power grids have significantly improved the efficiency and reliability of monitoring and controlling the physical power system \cite{fang2011smart}. 
Simultaneously, computer science and electronics techniques are enhancing cyber systems to improve the performance of computing and communication technologies. 
These parallel advances in both the physical power system and cyber systems are merging to form an integrated Cyber-Physical Power System (CPPS), which promises to revolutionize the management and operation of modern power grids \cite{sridhar2011cyber}.
CPPS is a novel system that integrates the internet and physical power system components, encompassing all aspects of electric power systems, including generation, transmission, distribution, and utilization \cite{yohanandhan2020cyber,huang2017smart}.
Nowadays, CPPS has been widely applied to various safety-critical power grid scenarios.
Especially in power grid safety control, CPPS has a faster processing capability to ensure the overall safe operation of the power grid through cyber network analysis. 
Correspondingly, this also presents additional requirements from both physical and cyber perspectives for deploying components in the CPPS \cite{lyu2019safety}.

Anomaly detection is essential to prevent potential system failures and ensure the safety of CPPS.
Recent progress in academia has significantly advanced the development of CPPS for anomaly detection.
Li et al. \cite{li2021dynamic} 
proposed an online anomaly detection method to accurately detect when an electrical component has failed based on a graph structure algorithm. 
With the rise of Artificial Intelligence (AI) technology, the CPPS anomaly detection framework has extensively incorporated machine learning and deep learning methods. For instance, Niu et at.
\cite{niu2019dynamic}
combined a long short term memory network and a convolutional neural network to develop a time-series
anomaly detector for data injection attack detection. 
For these anomaly detection methods,
efficient acquisition of real-time grid data requires advanced deployment of physical hardware.
Hooi et al. \cite{hooi2019gridwatch} proposed an approach for sensor placement to maximize the probability of detecting anomalies with limited equipment.

In addition to deploying hardware, it is critical for the cyber layer to maintain communication between edge devices during disruptions caused by various destruction, including physical failures, natural disasters, and malicious attacks, ensuring the proper functioning of the system.
The cyber layer's robustness typically depends on the topology structure of the network, which can be assessed using various metrics.
For instance, Schneider et al. \cite{schneider2011mitigation} proposed a robustness metric based on the percolation theory, which considers the maximal connected subgraphs after the repeated removal of the highest degree node. According to this robustness, Qiu et al. presented the ROSE, a robustness strategy for scale-free wireless sensor networks \cite{qiu2017rose}.
Moreover, Zhang et al. proposed the r-robustness of the networks guaranteeing connectivity even if some nodes are removed \cite{zhang2015notion}.

% Besides deploying hardware, efficient data aggregation algorithms enable real-time detection at the cyber layer. For example, Chiang \cite{chiang2014cycle} proposed a cycle-based data aggregation scheme for the smart grid, which links all cell heads to form a cyclic chain to collect data efficiently. 
Despite the significant progress, few papers investigate physical layer deployment and cyber layer communication connection simultaneously, i.e., \textit{co-design a secure sensor placement strategy for anomaly detection and a robust communication protocol}.
There are two difficult challenges that need to be addressed to resolve this problem
\begin{itemize}
    \item \textbf{Challenge 1}:
     Anomaly detection in CPPS needs to prioritize accuracy with limited resources \cite{yohanandhan2020cyber}. 
     First, a limited number of sensors should be placed at the most critical nodes to ensure effective data extraction.
    Moreover, dispersing sensors will significantly increase the monitoring of every part of the power grid, but the connectivity and robustness of the network will be limited due to long wireless communication links. On the other hand, concentrating sensors will ensure the network's connectivity but affect the detection accuracy \cite{xu2012efficient}.
    Consequently, an effective sensor placement must consider both the physical and cyber layers to find a trade-off.
    % significantly increase the time required for data aggregation, thus impacting the detection speed of the control center \cite{bagaa2014data}. On the other hand, concentrating sensors will interfere with data communications, affecting the quality of data and the accuracy of detection \cite{xu2012efficient}. 
    % Consequently, an effective sensor placement must consider both the physical and cyber layers to find a trade-off.
    \item \textbf{Challenge 2}:
    Sensor placement and robust network connection problems are typically challenging to solve. 
    It has been proven that many optimizations that both optimize the node selection and the robustness of networks are non-deterministic polynomial-time hardness (NP-hard) \cite{zhang2015notion}.
    Hence, many current solution algorithms are based on greedy algorithms or traditional optimization methods, which do not guarantee the optimal placement strategy \cite{clark2018greedy}. 
    Therefore, data-driven approaches, e.g., machine learning methods, become a preferred solution over mathematical model-based approaches to explore the global optimal solution.

    % the transmission structure for a given placement \cite{ganesan2006power} and relay node placement given communication protocol \cite{lloyd2006relay} are NP-hard.

% Also, existing solutions such as DRL are not scalable, e.g., 100 or 1,000 nodes. Therefore, data-based approaches (machine learning) becomes  a preferred solution over model-based approaches (math)

\end{itemize}

In this paper, we present an efficient framework for placing sensors for anomaly detection and guaranteeing robustness under link failures in CPPS.
First, we consider a one-to-one interdependent CPPS, where each node in the physical layer is controlled by one cyber node \cite{wu2021cyber}.
Next, we examine the edge robustness of the cyber layer, which refers to the ability to remain connected even when some edges are failed and disconnected.
To measure it, we introduce two metrics: the Cheeger constant and the Fiedler value \cite{de2007old}. These metrics are mathematical indicators in graph theory that measure the edge connectivity of a graph, which in turn helps assess the robustness of the cyber layer. For the physical layer, we utilize power detectors for anomaly detection introduced in GridWatch \cite{hooi2019gridwatch}. 
Based on considerations at both the cyber layer and the physical layer, we formulate the robust sensor placement problem as an optimization to maximize the robustness of the cyber layer while ensuring accurate and effective anomaly detection in the physical layer.
% Since there is always a communication delay, the data we collect pertains to a specific period. Instead of detecting anomalies at a single moment, detectors for time slots should be used.
% Thus,
% we propose time slot power detectors for anomaly detection based on GridWatch \cite{hooi2019gridwatch}. 
% Then, 
% we model the cyber layer of CPPS as a Wireless Sensor Network (WSN) with signal interference-plus-noise ratio (SINR) constraints to ensure communication quality \cite{du2022performance}. 
% Since the communication links between sensors are usually dynamic, we adopt a hierarchical
% data aggregation routing schedule
% to calculate a lower bound on communication time. 
% Based on the constraints, we formulate the secure sensor placement problem as an optimization to minimize data aggregation time while ensuring accurate and effective anomaly detection.
% Due to the NP-hardness of the proposed optimization, 
We utilize the Reinforcement Learning (RL) framework to efficiently find near-optimal solutions, which is suitable for the solution space that is vast and not easily navigable through traditional optimization techniques \cite{li2017deep}.
Additionally, diffusion-based policy RL algorithms have demonstrated state-of-the-art performance, particularly in network optimization contexts \cite{du2024enhancing}. The diffusion model employs a unique combination of diffusion and denoising processes to effectively explore the search space to navigate through potential solutions, progressively refining decisions through the denoising phase to converge on optimal strategies.
To solve the long time for the diffusion model convergence, we proposed
an Experience Feedback Graph Diffusion (EFGD) policy optimization approach to solve the proposed secure sensor placement optimization. 
The proposed framework utilizes reinforcement learning from human feedback (RLHF) to improve the convergence speed of training by leveraging prior exploration feedback, which informs the optimization strategy in moving closer to the optimal outcome \cite{du2024reinforcement}.
Our main contributions are summarized as follows.
\begin{itemize}
    \item 
    By using graph theory, we model
    a one-to-one interdependent CPPS \cite{kong2020routing}
    through the gird graph and communication graph with shared vertices, representing the potential location to place sensors.
    We introduce the Cheeger constant and the Fiedler value to measure the robustness of the cyber layer communications under link failures.
    Moreover, we utilize anomaly detectors in the physical layer to detect the abnormal power information of the power grid.
    % We present three time slot power detectors to detect the collected data over a period of time.
    % Moreover, we model the cyber layer of CPPS as a WSN, considering SINR constraints to ensure communication quality. Besides, we utilize a hierarchical
    % data aggregation routing schedule to calculate a lower bound on the number of time slots. 
    \textit{To the best of our knowledge, this is the first work that considers both the physical layer grid and the cyber layer network robustness simultaneously in CPPS for anomaly detection.
    } 

    \item 
    Based on the CPPS model, we formulate an optimization to maximize the robustness of the cyber layer network while ensuring accurate and effective anomaly detection in the power grid.
    % Based on the CPPS model, an optimization is formulated to minimize the detection time while ensuring the detection accuracy.
    We prove that the formulated problem is NP-hard, making it difficult to solve using existing optimization methods. 
    Inspired by the diffusion model-based optimization framework \cite{du2024enhancing}, we design an optimization framework utilizing graph diffusion to optimize the placement policy via the denoising process.

%     \todo{I think, it is good to add one short paragraph to explain why diffusion as there are other DRL and ML methods to solve NP-hard problems
% Otherwise, likely, reviewers will question}

    % adopt policy boosting to optimize the generation policy, allowing LEGD to reinforce itself through interacting with zero-trust networks.

    \item 
    Due to the long time for the diffusion model convergence, we propose the EFGD policy optimization algorithm, which adopts the cross-entropy gradient and introduces experience feedback into the training. Instead of using the negative log-likelihood gradient, the propped EFGD approach utilizes the cross-entropy gradient, allowing it to explore higher rewards and bring the predicted distribution closer to the latent distribution. Moreover, the inclusion of experiment feedback in the EFGD method can significantly improve the convergence during training.

\end{itemize}

The rest of the paper is organized as follows. 
Section \ref{sec:related} reviews related works.
The CPPS system model consisting of both physical layer and cyber layer is presented in Section \ref{sec:model}.
In Section \ref{sec:problem}, we first formulate the robust sensor placement optimization problem for anomaly detection.
Then, we elaborate on the design of EFGD to solve the optimization problem effectively.
Section \ref{sec:exp} provides and discusses the simulation results. Finally, Section \ref{sec:con} summarizes the paper.

% \begin{itemize}
%     \item \textbf{Motivation}:
%     1. Electrical Grid has the nature of a graph due to the connection between nodes and circuits. Using graph-based optimization algorithms is more helpful in solving sensor placement problems.

%     2. Most of the existing sensor placement algorithms are based on traditional optimization algorithms (which are difficult to solve nonlinear problems) or greedy algorithms (which are difficult to find the optimal solution). The above problems can be solved using reinforcement learning algorithms.

%     3. Some special properties of graph representation, such as non-differentiable, are difficult to optimize using general reinforcement learning algorithms. Algorithms based on graph generation, especially Graph Diffusion, have powerful performance for solving such graph optimization problems.

%     \item \textbf{Contribution}:
%     1. A node generation reinforcement learning algorithm based on graph diffusion is proposed to optimize the sensor placement problem.

%     2. The algorithm focuses on finding the optimal number and location of sensors simultaneously under a limited budget. Existing algorithms usually only concentrate on a given number of sensors.

%     3. The proposed algorithm achieves better performance than existing methods in experiments.

% \end{itemize}

% Interdependent Cyber-physical Power Systems Framework 

% % \subsection{Interdependent Cyber-physical Power Systems Framework}

\section{Background and Related Work}
\label{sec:related}

Interdependent CPPS frameworks are crucial for modernizing and enhancing power infrastructure robustness through the integration of computational and physical processes. A common application is the one-to-one interdependent CPPS, where each physical layer node in the physical layer is controlled by a cyber node \cite{wu2021cyber}.
The physical layer of the CPPS comprises power consumption and/or production nodes connected with lines. 
Some selected nodes in the physical layer function as a sensor that collects and monitors the status information, such as current and voltage. Through one-to-one interdependency, each sensor transmits the collected information to the corresponding information transmission unit. Subsequently, the transmission unit of each node sends the information to the unified power grid control center via the cyber layer's communication link. Finally, the power grid control center processes and analyzes the collected data and controls the CPPS based on the analysis results.

\begin{figure}[htbp]
    \centering
    \includegraphics[width= 0.95\linewidth]{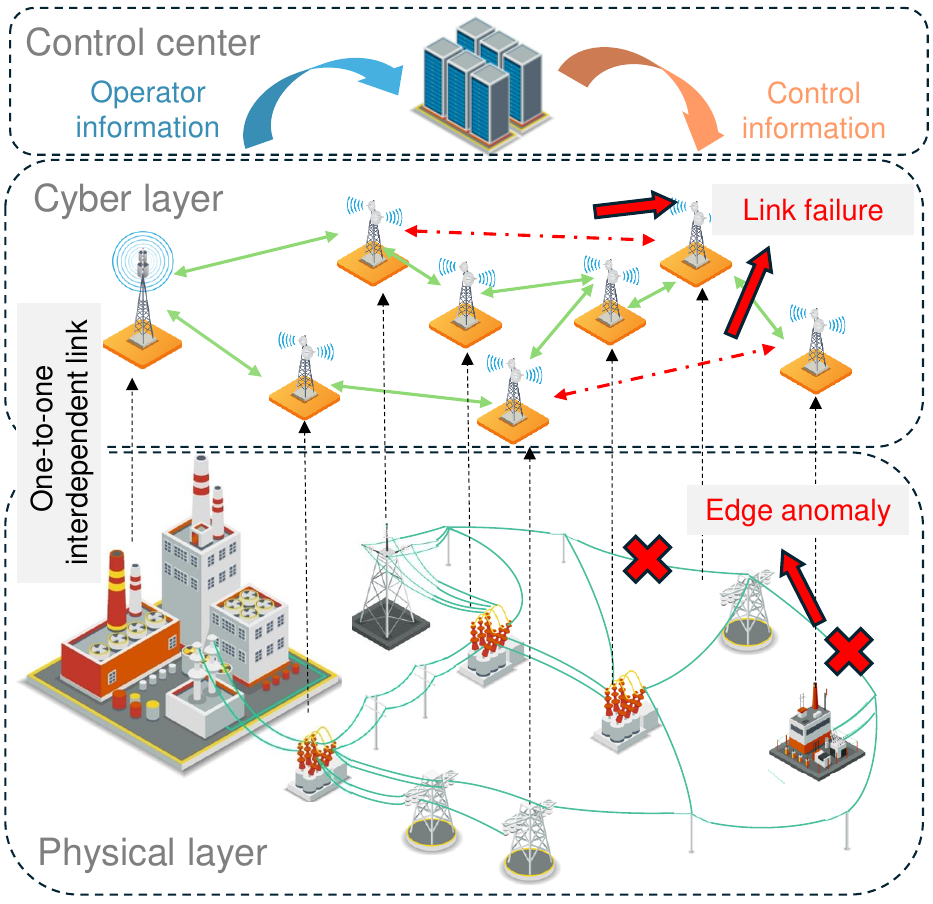}
    \caption{The system model of interdependent CPPS.
    \textit{Part A}. The illustration of the physical layer grid with edge anomaly. \textit{Part B}. The cyber layer network with link failures. \textit{Part C}. The control center processes data and controls the CPPS.
    }
    \label{fig:cpps}
\end{figure}

% a one-to-one interdependent CPPS, where each node in the physical layer is controlled by one cyber node \cite{wu2021cyber}. 
% The physical layer of the CPPS is an electrical grid connected by physical medium such as cables.
% Each node in the physical layer represents a sensor that collects and monitors the status information of the circuit connected to it, such as current and voltage. Via one-to-one interdependency, the sensor transmits the collected information to the corresponding information transmission unit. After receiving the information, the transmission unit of each node transmits the information to the unified power grid control center through the communication link of the cyber layer. Finally, the power grid control center will process and analyze the collected data, such as anomaly detection, and then control the CPPS based on the analysis results. 

\subsection{Anomaly Detection in Physical Layer}

Anomaly detection for the physical layer of CPPS aims to identify changes in the physical information of the circuit, such as voltage and current, to achieve grid monitoring.
Time series anomaly detection is a common method used in practice focusing on the temporal characteristics of circuit information \cite{keogh2007finding}.
For multivariate time series, 
there are numerous techniques to achieve detection, including convolutional neural networks \cite{yi2017grouped}, autoencoder models \cite{yang2025hybrid}, distance-based models \cite{ramaswamy2000efficient}, and isolation forests \cite{liu2008isolation}.
Additionally, temporal graph anomaly detection utilizes the topology structure 
to find anomalous changes in the grid, such as neighborhood-based  \cite{akoglu2010oddball} and community-based approaches \cite{chen2012community}. For dynamic graphs \cite{li2021dynamic}, the authors in \cite{akoglu2010event} found change points, while other researchers leveraged partition-based \cite{aggarwal2011outlier} and sketch-based \cite{ranshous2016scalable} approaches.
However, almost all these methods require fully observed data from deployed sensors and only seldom consider sensor selection or placement for effective detection \cite{hooi2019gridwatch}.
Moreover, the papers on sensor placement only consider one aspect of the grid system, such as state estimation \cite{li2011phasor}, without considering the transmission simultaneously.
% Moreover, several papers considering sensor placement \cite{dua2008optimal,li2011phasor} 
% do not focus on anomaly detection.
Motivated by this situation, we intend to provide a sensor placement approach for effective anomaly detection and enhance the transmission robustness simultaneously.

\new{

% \todo{rewrite why Fiedler}

\subsection{Robust Network in Cyber Layer}

The primary function of the CPPS cyber layer is to transmit, calculate, and collect data. Therefore, robust networks prioritize maintaining specific operations of an entire network, even in the face of disruptions such as link failures or node damage. 
Researchers usually use graphs to illustrate the topological structure of networks, utilizing various graph properties to verify and enhance network robustness \cite{koster2013robust}. 
Based on the percolation theory, the authors in \cite{schneider2011mitigation} proposed a robustness metric to measure network robustness under node failures. Additionally, by utilizing this robustness metric, a robust generation strategy for wireless communication networks is proposed, which significantly enhances the robustness of the network against nodes cyberattacks \cite{qiu2017rose, qiu2019robustness}.
Besides node failures, the authors in \cite{egeland2009availability} measured the availability of planned networks by analyzing the number of redundant nodes. This method enhances network reliability for any given topology by increasing redundant nodes. 
However, this approach is less applicable in resource-limited settings, as it improves reliability primarily by adding redundancy without accounting for strict deployment budget constraints.

% However, it cannot ensure that the total number of nodes remains below a specific level.

% Therefore, we utilize metrics, including the Cheeger constant and the Fiedler value in graph theory, to assess the robustness of a graph with link failures. These metrics are related to the structure of the graph rather than the number of nodes, allowing for a better trade-off between the number of nodes and the overall network structure \cite{de2007old}.

Therefore, we utilize spectral metrics, including the Cheeger constant and the Fiedler value, to assess network robustness under link failures~\cite{de2007old}.
Unlike degree-based or redundancy-based measures, these spectral quantities characterize the intrinsic connectivity of a graph and directly reflect the presence of bottlenecks and weakly connected components~\cite{chung1997spectral}.
In particular, the Cheeger constant quantifies the minimum cut of a network normalized by the node volume, providing a principled measure of how easily the graph can be disconnected~\cite{ghosh2006growing}.
Although computing the Cheeger constant exactly is NP-hard, Cheeger’s inequality establishes tight bounds between the Cheeger constant and the second smallest eigenvalue of the graph Laplacian, i.e., the Fiedler value, enabling efficient approximation with theoretical guarantees \cite{de2007old}.

Compared with other spectral properties such as the largest eigenvalue or spectral radius, which mainly capture global density or diffusion speed, the Fiedler value explicitly characterizes algebraic connectivity and is directly linked to edge-disconnection robustness~\cite{van2023graph}.
Moreover, maximizing the Fiedler value has been shown to improve resilience against link removals by promoting balanced connectivity across the network, rather than concentrating links locally~\cite{olfati2007algebraic}.
These properties make the Fiedler value particularly suitable for modeling robustness under link failures while remaining computationally tractable for large-scale graphs.
As a result, the combination of Cheeger’s inequality and the Fiedler value provides a theoretically grounded and practically efficient surrogate for robustness optimization, allowing a principled trade-off between network structure and deployment cost.

\subsection{Sensor Placement and Network Co-design}

The primary objective of sensor placement is to ensure accurate and timely acquisition of system states while operating under limited deployment budgets and communication resources~\cite{akyildiz2002survey, younis2008strategies}. Classical studies have formulated sensor placement as a coverage or observability problem, aiming to maximize sensing quality through carefully selected node locations~\cite{amundson2009survey}. Early studies have typically focused on single-layer optimization, where sensors are deployed to improve monitoring accuracy, state estimation, or anomaly detection performance without explicitly considering communication constraints~\cite{11311556}.
To address practical deployment requirements, various optimization-based and heuristic approaches have been proposed. For example, greedy and submodular methods have been widely adopted to maximize sensing coverage or detection probability under cardinality constraints~\cite{bhuiyan2014sensor}. However, these approaches often struggle with the vast solution spaces of complex power grids, frequently converging to local optima~\cite{sahu2022optimal}. In parallel, bio-inspired algorithms have been introduced to enhance scalability in large-scale networks~\cite{10.1016/j.jnca.2015.09.013}, yet they predominantly treat sensing and networking as decoupled components.

Recent studies have investigated the joint optimization of sensing and communication~\cite{iqbal2015wireless}. Such co-design frameworks aim to balance sensing accuracy and network lifetime~\cite{wang2011coverage}. Nevertheless, most existing solutions rely on redundant node deployment to improve robustness, which often leads to increased infrastructure costs and limited flexibility~\cite{sara2014routing}. Moreover, simple redundancy does not explicitly account for structural robustness under stochastic link failures, nor does it provide principled control over network topology.
Motivated by these limitations, we adopt a graph-theoretic perspective to jointly model sensor placement and network connectivity. Rather than merely increasing node density, we leverage spectral metrics, including the Cheeger constant and the Fiedler value~\cite{de2007old}, to characterize the robustness of the cyber layer. These measures enable a principled trade-off between sensing performance and structural robustness, providing a unified foundation for sensor placement and network co-design.

}

\new{

\subsection{RL-based Network Optimization}

Many network optimization problems in wireless and cyber-physical systems are inherently NP-hard due to their combinatorial decision spaces and coupled constraints~\cite{nguyen2018mobile}.
DRL has therefore been widely adopted as a data-driven alternative to conventional optimization methods, enabling agents to learn adaptive policies directly from environmental interactions~\cite{mao2018deep}.
Representative studies have applied DRL to power control, routing, and resource allocation, demonstrating its effectiveness in handling large state spaces and complex dynamics where analytical solutions are intractable~\cite{luong2019applications}.
Due to the topological nature of communication networks, network optimization problems are often represented as graphs.
To exploit such structural information, recent work integrates graph representations with DRL, enabling policies to capture node dependencies and improve generalization across different network instances~\cite{almasan2022deep}.
More recently, generative AI-based RL has emerged as a promising paradigm by directly modeling the distribution of optimal solutions rather than producing actions in a step-wise manner~\cite{wang2024generative2, zhao2025secdiff, wang2025generative2}.
Among them, diffusion models have attracted significant attention, as they learn data distributions by gradually injecting noise and recovering clean samples through a denoising process \cite{ho2020denoising}.
Owing to their strong representation capability, diffusion models have been widely applied to image, video, and audio generation \cite{yang2023diffusion}, and have recently been extended to decision-making by integrating the denoising process with DRL frameworks \cite{wang2024generative}.

In wireless networks, Du et al. \cite{du2024enhancing} proposed a diffusion model-based RL framework that formulates network control as a generative denoising process, achieving superior scalability and performance compared with conventional DRL methods \cite{zhao2024generative, sun2026generativeintentpredictionagentic}.
Motivated by the graph-structured nature of communication systems, Liu et al. \cite{liu2024graph} further introduced Graph Diffusion Policy Optimization (GDPO), which applies diffusion models to graph generation and enables constraint-aware network optimization \cite{wang2024empowering}.
However, limited by the slow convergence of diffusion models, these approaches typically require long training times to progressively refine noisy solutions toward high-quality graphs \cite{vecerik2017leveraging}.

Motivated by this limitation, in this paper, we propose a diffusion-based policy optimization framework with experience feedback, which explicitly incorporates high-reward trajectories into the denoising process.
By guiding each denoising step toward promising solution regions, the proposed approach significantly accelerates convergence while preserving the global exploration capability of diffusion models.

}
\section{System Model: Interdependent Cyber-Physical Power Systems
for Anomaly Detection}
\label{sec:model}

In this section, we present our interdependent CPPS model designed for accurate anomaly detection and robust data transmission.
\new{
For reader convenience, Table~\ref{tab:notation} summarizes the key symbols and parameters used throughout this paper.
}

\begin{table*}[t]
\centering
\caption{\new{Key notations and parameters.}}
\label{tab:notation}
\renewcommand{\arraystretch}{1.1}
\setlength{\tabcolsep}{6pt}
\begin{tabular}{ll ll}
\toprule
\textbf{Symbol} & \textbf{Description} &
\textbf{Symbol} & \textbf{Description} \\
\midrule

$A(k,M)$ & Overall anomaly score
& $a_i(k)$ & Node anomaly score \\

$BP L(d_0)$ & Reference path loss
& $\beta$ & Feedback weight \\

$D$ & Distance matrix
& $\Delta S_{i,e}(k)$ & Power variation \\

$E_C$ & Candidate communication links
& $G_C$ & Cyber sensor network graph \\

$G_P$ & Physical power grid graph
& $h(L_{G_C})$ & Cheeger constant \\

$I_e(t)$ & Line current
& $l_{i,j}$ & Link indicator \\

$L_{G_C}$ & Graph Laplacian
& $\lambda_2$ & Second smallest eigenvalue \\

$\lambda_a$ & Anomaly threshold
& $\lambda_c$ & Communication threshold \\

$\lambda_s$ & Safety threshold
& $N$ & Sensor budget \\

$N_t$ & Training iterations
& $P_n$ & Noise power \\

$PL(v_i,v_j)$ & Path loss between nodes $v_i$ and $v_j$
& $P_t(v_i)$ & Transmit power \\

$p_\theta$ & Denoising network
& $Q_t^V,Q_t^E$ & Diffusion transitions \\

$r(G_0)$ & Final reward
& $r_1,r_2,r_3$ & Reward weights \\

$S=(V_S,E_S)$ & Sensor placement solution
& $S_a(V_S)$ & Detection score \\

$S_{i,e}(t)$ & Edge power measurement
& $SNR(v_i,v_j)$ & Link SNR \\

$V$ & Node set
& $V_i(t)$ & Node voltage \\

$X_i(k)$ & Combined detector output
& $X_\sigma$ & Shadowing term \\

$x_{GA,i}$ & Group-anomaly detector
& $x_{GD,i}$ & Group-diversion detector \\

$x_{SE,i}$ & Single-edge detector
& $|\mathcal{B}|$ & Replay buffer size \\

$|\mathcal{D}|$& Trajectory set size
& $\gamma$ & Path-loss exponent \\

\bottomrule
\end{tabular}
\end{table*}

% \todo{Distinguishing between adjacency matrices and sets}
% \todo{Symbol and naming consistency}
% \todo{using bm represents verctor}

% we may put a brief description of the physical- and cyber-layer before the graph models.

\subsection{Interdependent Cyber-Physical Power System Model}

The proposed CPPS consists of a physical layer, which is a transmission grid, and a cyber layer modeled by a wireless sensor network (WSN).
As shown in Fig. \ref{fig:cpps}, we consider a one-to-one interdependent CPPS model \cite{wu2021cyber}. 
The physical layer is formally represented by a graph $\mathcal{G}_P = (\mathcal{V}_P, \mathcal{E}_P)$, where $\mathcal{V}_P$ denotes the set of grid nodes in which the sensors can be placed, and $\mathcal{E}_P$ is the set of nodes connectors (electrical lines), as shown in Fig. \ref{fig:cpps} \textit{Part A}.
% and $V_P$ and $E_P$ represent the corresponding vertex matrix and adjacency matrix.
Similarly, we model the WSN cyber layer by another graph $\mathcal{G}_C = (\mathcal{V}_C, \mathcal{E}_C)$, where $\mathcal{V}_C$ represents the set of potential locations for data transmission units, and $\mathcal{E}_C$ denotes the set of corresponding communication links (Fig. \ref{fig:cpps} \textit{Part B}). 
To achieve effective one-to-one interdependent transmission, we assume the gird sensor in the physical layer and its corresponding data transmission unit in the cyber layer are placed on the same node. 
Thus, the choice of nodes where sensors may be deployed should also be the same, i.e., $\mathcal{V}_P = \mathcal{V}_C$. For convenience, we use $\mathcal{V}$ jointly to represent this set of vertices where sensors can be placed. 
Moreover, we consider one power grid control center in the CPPS, which monitors the operation of the CPPS and sends control instructions
through the connection with the cyber layer.
(Fig. \ref{fig:cpps} \textit{Part C}).

\subsection{Cyber Layer Model for Robust Communication}

In this subsection, we will present our cyber layer model for robust communication.

\subsubsection{Communication Link}

We first consider reliable communication links between access nodes within the cyber layer, which can be measured by various metrics. In this paper, we utilize the Log-Normal Shadowing Path Loss (LNSPL) model to assess the communication quality of links, which can be generalized to various environments, including indoor, inter-vehicular, and near-ground scenarios \cite{kurt2017path}.

% \todo{give examples}

In the cyber layer $\mathcal{G}_c=(\mathcal{V}, \mathcal{E}_C)$, we use a matrix $D=\{d_{i,j}\}$ to denote the distance between each node, where $d_{i,j}$ represents the physical distance between nodes $v_i$ and $v_j$ in the set of vertices $\mathcal{V}$.
Accordingly, the path loss $PL$ of the communication link between $v_i$ and $v_j$ can be expressed as~\cite{sandoval2017improving}
\begin{equation}
\label{eq:lnspl}
PL(v_i,v_j) = BPL(d_0) + 10\cdot \gamma\cdot\log_{10}(\frac{d_{i,j}}{d_0})+X_{\sigma},
\end{equation}
where $BPL(d_0)$ represents the reference path loss of distance $d_0$, $\gamma$ is the path loss exponent reflecting the rate at which the signal attenuates with distance, and $X_\sigma$ is a zero mean Gaussian random variable with the variance of $\sigma$, denoting the shadow fading.
In this model, 
the signal-to-noise ratio $SNR(v_i,v_j)$ from the transmitter $v_i$ is given by
\begin{equation}
\label{eq:snr}
    SNR(v_i,v_j) = P_t(v_i) - PL(v_i,v_j) - P_n,
\end{equation}
where $P_t$ and $P_n$ represent the transmit power and noise power. Since we are in the placement stage, we assume that the transmission power of each sensor consistently reaches the minimum $P_t^m$ and the noise power remains uniform across all communication links at the worst case $P_n^M$.
Consequently, the SNR we calculated is based on the worst transmission conditions, serving as a lower bound to ensure safe operation after placement.
Under this assumption, to maintain reliable communication links, the communication link is active between two nodes only if the value of the signal-to-noise ratio is high, i.e., the path loss does not exceed a certain threshold.
Given a threshold $\lambda_c$, the communication state $l_{i,j}$ between nodes $v_i$ and $v_j$ can be written as
\begin{equation}
\label{eq:linkcons}
    l_{i,j}=
    \begin{cases}
        1, ~~~PL(v_i,v_j)\leq\lambda_c\wedge(v_i,v_j) \in\mathcal{E}_C,\\
        0, ~~~\text{otherwise,}     
    \end{cases}
\end{equation}
where $l_{i,j} = 1$ indicates the activated link, $l_{i,j} = 0$ represents the link is not activated, and $(v_i,v_j) \in\mathcal{E}_C$ indicates the link is in the set of potential communication links, as demonstrated in Fig. \ref{fig:pattern} \textit{Part B}.

% LNSPL（Log-Normal Shadowing Path Loss）

% https://arxiv.org/pdf/1603.08014

% commnucaiton range

\subsubsection{Robust Network}

In this part, we present several critical metrics for evaluating the robustness of communication networks.
We consider the robustness of the cyber layer to ensure that network communication remains reliable and functional. Specifically, we measure the robustness of the network as a capability to maintain connectivity even when some of its links fail 
\cite{liu2022network}.
In this paper, we leverage the Cheeger constant of the graph as the performance indicator. The Cheeger constant, also known as the Isoperimetric number, measures the weak connections in a graph, where a higher Cheeger constant indicates better connectivity and fewer bottlenecks \cite{mohar1989isoperimetric}.

Considering the cyber layer graph $\mathcal{G}_C = (\mathcal{V}, \mathcal{E}_C)$, $E_{C} = \{\alpha_{i,j}\}_{n \times n}$ is a 0-1 adjacency matrix of $\mathcal{G}_C$, where $n$ represents the number of vertices, i.e., $n=\vert\mathcal{V}\vert$, and $\alpha_{i,j}$ equals 1 if vertices $v_i$ and $v_j$ are connected, and 0 otherwise.
$D = diag\{\beta_{1},\ldots,\beta_{n}\}$ is the degree matrix of the cyber layer, where $\beta_{i} = \sum_{1\leq j\leq n, j\neq i}\alpha_{i,j}$ represents the degree of vertex $v_i \in \mathcal{V}$.
Given the adjacency matrix $E_{C}$ and the degree matrix $D$, the Laplacian matrix of graph $\mathcal{G}_C$ can be expressed as $L_{\mathcal{G}_C} = D - E_C$.
For the normalized Laplacian matrix
$\mathcal{L_{\mathcal{G}_C}}=D^{-1/2}L_{{\mathcal{G}_C}} D^{-1/2}$,
the Cheeger constant in spectral graph
theory is defined as 
\cite{chung1997spectral}:
\begin{equation}
    h(\mathcal{L}_{{\mathcal{G}_C}}) = \min_{\mathcal{Y}}\frac{\sum_{i\in \mathcal{Y}, j \in \overline{\mathcal{Y}}} \alpha_{i,j}}{\min\{vol(\mathcal{Y}), vol(\overline{\mathcal{Y}})\}},
\end{equation}
where $h(\mathcal{L}_{{\mathcal{G}_C}})$ indicates the Cheeger constant of graph $\mathcal{L}_C$, $\mathcal{Y}\subset \mathcal{V}$ is a subset of the nodes and $vol(\mathcal{Y})=\sum_{i\in \mathcal{Y}}\beta_i$ denotes the volume of node set $\mathcal{Y}$.
However, the computation of the Cheeger constant for a given graph is NP-hard, leading to a substantial computational complexity in the
subsequent sensor placement \cite{yoshida2019cheeger}.

Therefore, we utilize the bounds of the Cheeger constant to provide an approximate guarantee, which can be computed by the Cheeger's inequality.
The Cheeger's inequality can be written as \cite{de2007old}:
\begin{equation}
    \lambda_2(\mathcal{L}_{{\mathcal{G}_C}})/2 \leq h(\mathcal{L}_{{\mathcal{G}_C}}) \leq \sqrt{2\lambda_2({\mathcal{L}_{{\mathcal{G}_C}}})},
\end{equation}
where $\lambda_2(\mathcal{L}_C)$ denotes the second smallest eigenvalue of the network Laplacian matrix $\mathcal{L}_C$.
Let $\bm{0}$ and $\bm{1}$ represent the vectors with all coordinates equal to 0 and 1, respectively, 
the eigenvalue can be computed by \cite{de2007old}:
\begin{equation}
\label{eq:fisher}
    \lambda_2(\mathcal{L}_{{\mathcal{G}_C}}) = \min_{v\neq\bm{0},v\perp \bm{1}} \frac{\langle\mathcal{L}_{{\mathcal{G}_C}} v,v\rangle}{\langle v,v\rangle}.
\end{equation}
Additionally, the second smallest eigenvalue $\lambda_2(\mathcal{L}_{{\mathcal{G}_C}})$ is referred to as the Fiedler value of the graph, which has a specific connection to the connectivity of the graph \cite{de2007old}, as shown in Fig. \ref{fig:pattern} \textit{Part C}.

% \todo{
% 为什么用 Cheeger + Fiedler（二阶特征值）

% justify why this specific approach was chosen

% 这是审稿人最“数学”的问题。

% 你现在只写了 Cheeger inequality，没有解释为什么不用别的谱量。
% }

% The algebraic connectivity of the graph. This eigenvalue is also called the Fiedler value of the graph, named after Charles Fiedler, who discovered a direct link between this property and the connectivity of the graph.

% Via limiting the lower bound $\lambda_2(\mathcal{L})/2$, we can ensure the Cheeger constant $h(\mathcal{L})$ exceeds a certain threshold, thereby ensuring the resilience of the communication network.
% Thus, given a threshold $\lambda_C$, the resilience constraint can be expressed as
% \begin{equation}
% \label{eq:consinr}
%     SINR_{i,j}^{k} \geq \lambda_C\cdot l_{i,j}^{k}, 1\leq k\leq K, 1\leq i,j \leq n.
% \end{equation}

% https://arxiv.org/pdf/1901.02613

% Cheeger constant

% SageMath 

% https://arxiv.org/pdf/2009.12738

\subsection{Physical Layer Model for Anomaly Detection}

In this subsection, we will introduce our detailed physical layer model designed for anomaly detection.

\subsubsection{Sensor Information Collection}

When an anomaly occurs in the physical layer of CPPS, such as a transmission line failing or a grid component failing, the voltages at some nodes and the currents along the edges of the grid will change \cite{eiteneuer2019dimensionality}.
By placing sensors on the grid nodes, the voltage and current changes of the grid can be measured.
% Based on the sensor placement, anomaly detection aims to analyze the collected grid information to determine if the current grid has a fault. 

To be more specific, consider 
% a CPPS $\mathcal{X}$ and its 
the physical layer graph $\mathcal{G}_P = (\mathcal{V}, \mathcal{E}_P)$.
For a grid sensor deployed at vertex $v_i\in\mathcal{V}$, 
it can measure the voltage $V_i(t)\in\mathbb{C}$ of node $v_i$ at time $t$. Furthermore, it can also measure the current $I_e(t)\in\mathbb{C}$ along the edge $e\in\mathcal{N}_i$, where $\mathcal{N}_i \subseteq \mathcal{E}_P$ denotes the set of edges adjacent to the vertex $v_i$.
Compared with directly analyzing the changes in current and voltage, the power of the node combines the characteristics of both information. Therefore, the state of power can simultaneously indicate the changing characteristics of current and voltage, providing 
better anomaly detection in practice \cite{park2024anomaly}.
With the detected node voltage $V_i(t)$ and the edge current $I_e(t)$, the complex power along the edge $e$ is 
\begin{equation}
    S_{i,e}(t) = V_i(t)\cdot I_e(t)^{*},
\end{equation}
where $^*$ is the complex conjugate. 
In two consecutive detections, we denote the change in power along the edge $e$ detected by sensor $i$ as $\Delta S_{i,e}(t)$.
% Since there exists a delay of sensor data transmission in CPPS communication, the data for each time slot is usually transmitted together instead of being sent at each moment.

% To analyze the power changes of each time slot,
% we denote the maximal change in power along the edge $e$ detected by the sensor $v_i$ in $k$-th time slot as
% \begin{equation}
%         \Delta S_{i,e}(k)=\max_{t_{s}^{k} \leq m,n \leq t_e^k} \vert S_{i,e}(m)-S_{i,e}(n) \vert,
% \end{equation}
% where $t_s^k$ and $t_e^k$ represent the starting and ending times of the $k$-th time slot, respectively.
% let $\Delta S_{ie}(t)$ denotes the change in power along the edge $e$ detected by sensor $i$,

% \todo{add sampling rate and interval}
% \pending{In this situation, letting $\Delta S_{ie}(t)$ denotes the change in power along the edge $e$ detected by sensor $i$, i.e., $\Delta S_{ie}(t)=S_{ie}(t)-S_{ie}(t-1)$}

\subsubsection{Anomaly Detector}

Anomaly can cause the voltage and current on one side to surge or decrease, thus producing complex effects on the entire grid. 
% We consider a power grid consisting of a single generator, a single load, and power lines with uniform resistance, as shown in Fig. \ref{fig:pattern}. 
When one of the edges fails, 
the current will be redistributed among the edges of the grid as shown in Fig. \ref{fig:pattern} \textit{Part B}.
According to the current diversion, this redistribution leads to three anomaly patterns: single-edge anomaly, group anomaly, and group-diversion anomaly \cite{hooi2019gridwatch}.

Based on the characteristics of the three types of anomalies, the corresponding three detectors are defined as follows:

\begin{itemize}
    \item \textbf{Single-Edge Detector}:
    This detector focuses on the largest absolute change in power in the edges adjacent to a sensor $v_i$ in the $k$-th time, i.e.,
    \begin{equation}
    \label{eq:det1}
        x_{SE,i}(k) = \max_{e\in\mathcal{N}_i} |\Delta S_{i,e}(k)|.
    \end{equation}
    \item \textbf{Group Anomaly Detector}:
    This detector calculates the sum of all power changes in the edges adjacent to a sensor $v_i$, i.e.,
    \begin{equation}
    \label{eq:det2}
        x_{GA,i}(k) = |\sum_{e\in\mathcal{N}_i}(\Delta S_{i,e}(k))|.
    \end{equation}
    \item \textbf{Group-Diversion Detector}:
    The last detector computes the total absolute deviation of power changes about sensor $v_i$, i.e.,
    \begin{equation}
    \label{eq:det3}
        x_{GD,i}(k) = \sum_{e\in\mathcal{N}_i}|\Delta S_{i,e}(k)-\mathop{\Mean}\limits_{e'\in\mathcal{N}_i}(\Delta S_{i,e'}(k))|.
    \end{equation}
\end{itemize}

According to the results obtained by three detectors following Eqs. \eqref{eq:det1}-\eqref{eq:det3}, we define the total slot detector in at time $k$ as a vector $X_i(k) = [x_{SE,i}(k)~x_{GA,i}(k)~x_{GD,i}(k)]$, combining the outputs of three detectors simultaneously.

\subsubsection{Anomaly Score}

Based on the total detector result $X_i(k)$, 
we define an anomaly score to determine whether there is a fault in the grid.
Considering sensor $v_i$ and the total detector result $X_i(k)$, the anomaly score of this sensor in the $k$-th time can be computed via sensor-level anomalousness defined as follows \cite{hooi2019gridwatch}:
\begin{equation}
\label{eq:sensor-level}
    a_i(k) = \Big\Vert \frac{X_i(k)-\widetilde{\mu}_i(k)}{\widetilde{\sigma}_i(t)}\Big\Vert _{\infty},
\end{equation}
where $\widetilde{\mu}_i(t)$ and $\widetilde{\sigma}_i(t)$ are the historical median and inter-quartile range (IQR) \cite{yule1927introduction} of $X_i(t)$, respectively; the infinity-norm $\Vert\cdot\Vert_\infty$ represents the maximum absolute value.

% \todo{modify according to time slot}
% \pending{Utilizing this sensor level anomaly score, We can judge the failure at a certain time $t$. However, the control center usually needs to process for a period of time rather than at a certain moment. Therefore, 
% we define a period anomaly score.......}

Moreover, considering only the abnormal score of one sensor measurement cannot reflect the overall abnormal situation since the power of the electrical grid changes dynamically with time. Thus, we define the overall anomaly score of a set of nodes $\mathcal{M}\subseteq\mathcal{V}$ as the maximum sensor-level anomaly score in $\mathcal{M}$, as follows: 
\begin{equation}
    A(k, \mathcal{M}) = \max_{v_i\in\mathcal{M}} a_{i}(k).
\end{equation}
If the overall anomaly score exceeds a given threshold $\lambda_a$, i.e., $A(k,\mathcal{M})>\lambda_a$, the sensors detected an anomaly within this time slot.

For a sensor placement performance evaluation, we need to consider its ability to detect abnormalities within a series of time slots.
Specifically, suppose in a series of time $t_s = \{t_1,t_2,\ldots, t_T\}$, $s$ anomalies occurred at time $R_a = \{r_1,r_2,\ldots,r_s\}$. For the set of nodes $\mathcal{M}$, the anomaly detection score $S_a$ is 
\begin{equation}
\label{eq:anomaly score}
    S_a(\mathcal{M}) = \frac{1}{s}\sum_{i=1}^{s}\mathbb{I}(A(r_i,\mathcal{M})>\lambda_a),
\end{equation}
where $\mathbb{I}(\cdot)$ is the indicator function, and $\lambda_a$ is a given threshold. 

\section{Wireless Sensor Placement via Graph Diffusion Policy Optimization}
\label{sec:problem}

% \todo{Name: Dense Reward Graph Diffusion Policy Optimization (EFGDPO)}

% Via limiting the lower bound $\lambda_2(\mathcal{L})/2$, we can ensure the Cheeger constant $h(\mathcal{L})$ exceeds a certain threshold, thereby ensuring the resilience of the communication network.
% Thus, given a threshold $\lambda_C$, the resilience constraint can be expressed as
% \begin{equation}
% \label{eq:consinr}
%     SINR_{i,j}^{k} \geq \lambda_C\cdot l_{i,j}^{k}, 1\leq k\leq K, 1\leq i,j \leq n.
% \end{equation}

In this section, we formulate the problem of robust wireless sensor placement in CPPS. We then introduce the EFGD policy optimization algorithm, which utilizes the diffusion framework for graph generation and optimizes policies through graph diffusion.

% \todo{add matix}

\begin{figure}[htbp]
    \centering
    \includegraphics[width= 0.9\linewidth]{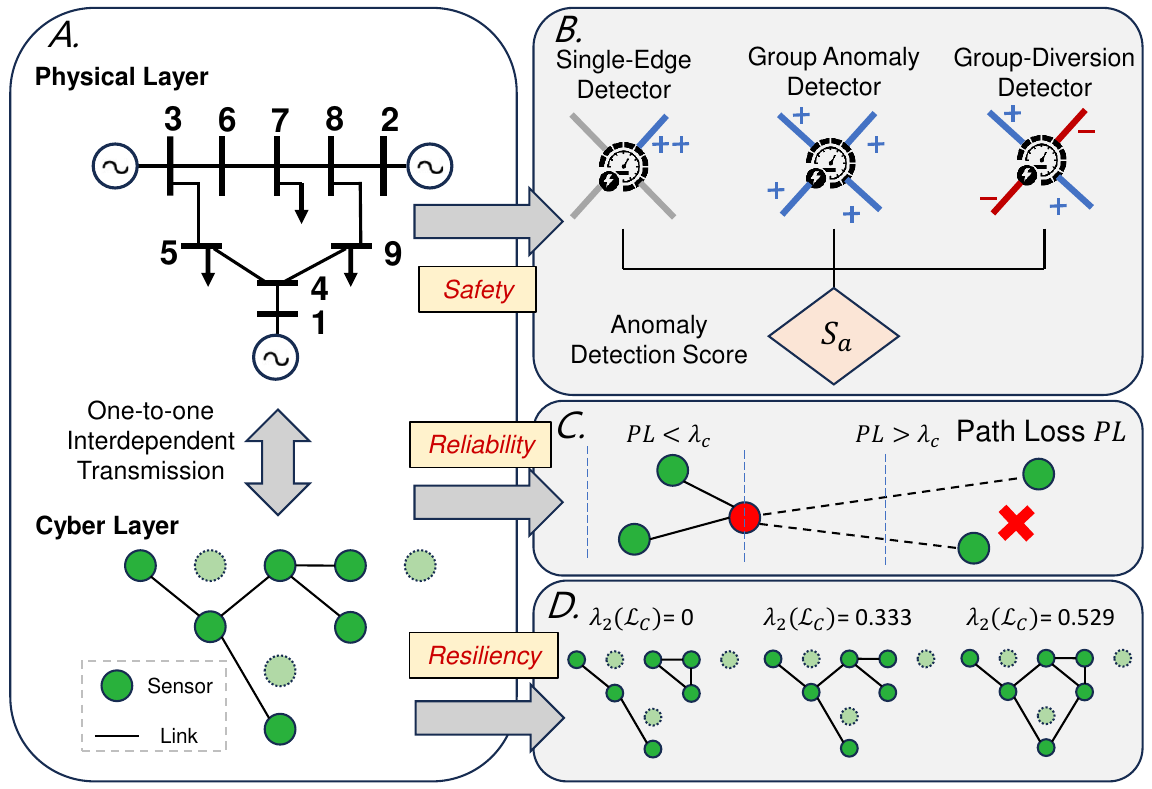}
    \caption{
    The illustration of the optimization objections. \textit{Part A}. The proposed CPPS system based on 
    IEEE 9 bus system. \textit{Part B}. The anomaly detectors focusing on the safety of the system. \textit{Part C}. The communication quality based on LNSPL model targeting the reliability. \textit{Part D}. The Fiedler value of different graphs addressing the robustness of the system.}
    \label{fig:pattern}
\end{figure}

\subsection{Problem Formulation}

% https://arxiv.org/pdf/2005.04585

% Dynamic Interference Management for UAV-Assisted Wireless Networks

% Dynamic Adaptive Anti-Jamming via
% Controlled Mobility

% \todo{- for the sensor placement decision $\mathcal{S}=(V_S, E_S)$, $E_S$ seems not necessary, I mean, if S is determined, according to the pass loss and the fixed geographic location, E will be fixed, but correct me if I'm wrong. this influences the constraint 13c of the optimization problem. We may need more explanation here, to emphasize the decisions include both nodes and links selection, thus it's a graph optimization.}

According to the system model defined in Section \ref{sec:model}, we formulate the robust sensor placement problem
considering both the physical layer layout and the
cyber layer links simultaneously. Specifically, we aim to optimize a robust sensor placement $\mathcal{S} = (\mathcal{V}_S, \mathcal{E}_S)$, where $\mathcal{V}_S\subseteq\mathcal{V}$ consists of nodes to place sensors, and $\mathcal{E}_S\subseteq\mathcal{E}_C$ indicates the activated links. The objective is to maximize the accuracy of anomaly detection and ensure the network's robustness.

Specifically, the anomaly detection accuracy can be computed by the anomaly detection score $S_a$ in Eq. \eqref{eq:anomaly score}.
According to the LNSPL model in Eq. \eqref{eq:lnspl}, 
the adjacency matrix $E_C$ is composed of the communication state $l_{i,j}$ in Eq. \eqref{eq:linkcons} to guarantee the quality of communication when transmitting operation and control information.
Additionally, the path loss in Eq. \eqref{eq:lnspl} varies over time due to shadowing. We assess the effectiveness of sensor placement under diffusion communication conditions to ensure robustness across different shadowing scenarios.
Moreover, via maximizing the lower bound, the Fiedler value $\lambda_2(\mathcal{L}_{\mathcal{S}})$, computed by Eq. \eqref{eq:fisher}, 
we can approximate the Cheeger constant $h(\mathcal{L}_{\mathcal{S}})$, thereby ensuring the robustness of the communication network.
In summary, our goal is to maximize
the robustness of the cyber layer while ensuring accurate and
effective anomaly detection of CPPS.
% as mentioned in Eq. \eqref{eq:anomaly score}, the sensor placement $\mathcal{S}$ needs to ensure the accuracy of CPPS anomaly detection. Similarly, according to SINR in Eq. \eqref{eq:consinr}, we expect that sensors can guarantee the quality of communication when transmitting collected data.
% Moreover, to ensure information freshness, the number of time slots should be minimized. 
% According to the hierarchical data aggregation routing schedule, the number of time slots is the height $h$ of sensor placement (data aggregation tree) $\mathcal{S}$. 
% Therefore, our goal is to ensure accurate communication, satisfying the constraints in Eq. \eqref{eq:anomaly score} and Eq. \eqref{eq:consinr}, while minimizing the number of time slots, i.e., the tree height $h$.
\new{
The robust sensor placement problem can be formulated as:

\begin{subequations}
\label{eq:optimization}
\begin{align}
    &~~~~~~\underset{\mathcal{S}}{\max} \quad \lambda_2(\mathcal{L}_{\mathcal{S}})  && \label{eq:maximize} \\
& \text{s.t.}~~~S_a(\mathcal{V}_S)\geq \lambda_s, && \label{eq:T_constraint} \\
& ~~~~~~PL(v_i,v_j)\leq\lambda_c, ~~~(v_i,v_j)\in\mathcal{E}_S, && \label{eq:l_constraint}\\
& ~~~~~~\vert\mathcal{V}_S\vert\leq N, && \label{eq:O_C_3}
\end{align}
\end{subequations}
% \[l_{i,j}=
%     \begin{cases}
%         1, ~~~PL(v_i,v_j)\leq\lambda_C\wedge(v_i,v_j) \in\mathcal{E}_C,\\
%         0, ~~~\text{else,}     
%     \end{cases}\]
% \begin{subequations}
% \label{eq:optimization}
% \begin{align}
%     &\underset{G}{\text{maximize}} \quad h(G) && \label{eq:maximize} \\
% & \text{s.t.}~~~PL(d)\geq \lambda_A, d\in E&& \label{eq:T_constraint} \\
% & ~~~~~~S_a(\mathcal{N}) \geq \lambda_B, && \label{eq:l_constraint}\\
% & ~~~~~~I_{tree} (\mathcal{T}, v_1) = 1, && \label{eq:Tree_constraint} 
% \end{align}
% \end{subequations}
% where $\lambda_s$ and $\lambda_c$ are two given threshold values to ensure safety and reliability, respectively; $N$ is the maximal number of sensors that can be placed.
where $\lambda_s$ and $\lambda_c$ are threshold values determined by the specific safety standards (e.g., detection rate) and communication hardware specifications (e.g., minimum SNR), respectively; $N$ is the maximum number of sensors constrained by the deployment budget. These parameters can flexibly be configured to adapt to diverse operational scenarios and environmental conditions. Detailed settings are discussed in Section~\ref{sec:exp_set}. Notably, Eq.~\eqref{eq:optimization} formulates a deterministic snapshot-based optimization to ensure tractability. 
However, by employing Monte Carlo simulations~\cite{harrison2010introduction} during the solving and verification process, we ensure that the derived placement strategy achieves robustness across the majority of stochastic shadowing scenarios.
}

% $h(\mathcal{S})$ indicates the height of tree $\mathcal{S}$, and $I_{tree}(\mathcal{S}, v_1)$ is a judgment function that equals 1 when a graph $\mathcal{S}$ is a tree with the root $v_1$, otherwise 0.
% In the above formulation, Equations (2) and (3) specify the range of transmitting power and the traffic demands of all nodes, respectively; Equations (6), (7), (8) and (11) are data gathering, link scheduling and SINR threshold constraints; Equations (9) and (10) are used to determine the received power level and the SINR of each link, respectively. Note that the allocated time slots in the optimal result may be inconsecutive, e.g., time slots, 1, 2, 5 and 7 are allocated, while time slots 3, 4 and 6 are not allocated. These unallocated time slots can be removed and the remaining time slots form a time frame.
% \begin{remark}
%     The structure constraint in Eq. \eqref{eq:Tree_constraint} can be solved efficiently by Depth-First Search or Breadth-First Search with a time complexity of $O(n)$ \cite{tarjan1972depth}.
% \end{remark}
% In summary, this optimization problem can also be seen as a graph generation problem. 

Although we approximate the Cheeger constant to simplify the constraints, we have the following proposition regarding the NP-hardness of the above optimization problem.
\begin{prop}
\label{prop:nphard}
    The proposed robust sensor placement optimization problem \eqref{eq:optimization} is NP-hard.
\end{prop}
\begin{proof}
To prove the NP-hardness of the proposed problem, we reduce the maximum algebraic connectivity augmentation problem, a well-known NP-hard problem \cite{mosk2008maximum}, to the formulated problem.
Specifically, for a given graph $\mathcal{G}=(\mathcal{V},\mathcal{E})$ in the maximum algebraic connectivity augmentation problem, 
we set the anomaly score $S_{a}(\mathcal{V}) = \lambda_s$ and $S_{a}(\mathcal{V}-{v_i})<\lambda_s$, $\forall v_i\in\mathcal{V}$, i.e., the constraint in Eq. \eqref{eq:T_constraint} is satisfied if and only if all vertices are selected. Moreover, we assign the weight of each edge in $\mathcal{E}$ satisfying the constraint in Eq. \eqref{eq:l_constraint}, while the edges not included in $\mathcal{E}$ dissatisfy the constraint.
Via this mapping, if we have an algorithm that can effectively solve the proposed robust sensor placement problem, it can solve the maximum algebraic connectivity augmentation problem.
Therefore, the solution to the maximum algebraic connectivity augmentation problem can be derived through 
the solution of the proposed optimization problem \eqref{eq:optimization}.
Since the maximum algebraic connectivity augmentation problem is NP-hard \cite{mosk2008maximum}, the proposed optimization problem is NP-hard.
\end{proof}

% However, an optimal schedule for general graph topology
% is difficult to obtain. This is due to the following: (1)
% the number of possible routing trees grows exponentially
% with the size of the graph and (2) the transmission schedule
% is subject to more constraints for a general graph than for
% line or tree topology. In fact, we prove that finding the
% optimum schedule for a general graph is NP-hard. The
% proof is constructed by reducing from a partition problem
% (a known NP hard problem) to a simplified graph scheduling
% problem.

Solving NP-hard problems using traditional optimization methods is extremely challenging due to the significant computational resources and unpredictable solving times required. These methods often struggle with large and complex solution spaces that are not easily navigable using conventional techniques \cite{li2017deep}, where RL excels.
Furthermore, diffusion-based RL algorithms can effectively explore the solution space through diffusion and denoising processes, allowing them to search for optimal strategies. These algorithms have demonstrated state-of-the-art performance, particularly in graph generation tasks \cite{liu2024graph}.
Therefore, we present our proposed EFGD (Fig. \ref{fig:framework}) to solve such a graph generation problem.

% \todo{again, explain why this graph diffusion, especially compare with other methods...}

% \todo{one algorithm for the whole framework}

% \todo{one algorithm for modified graph diffusion}

% \todo{x is determined by distance or other things}

\subsection{Graph Diffusion}

In this subsection, we will introduce the EFGD, which uses the generative diffusion model \cite{ho2020denoising} to generate a solution graph for optimizing sensor placement.

% We present LEGD, which leverages the diffusion model [17] to achieve strong generation ability and employs an LLM-empowered agent to enhance the capability of topology understanding. Next, we demonstrate the core of LEGD, i.e., diffusion on graph-structured data.

% \subsubsection{Graph Representation}

% In most of the diffusion models processing tasks, such as image generation, video generation, and audio generation, the data processed is generally continuous \cite{yang2023diffusion}.

\subsubsection{Discrete Diffusion Model}

Generative diffusion models, inspired by non-equilibrium thermodynamics, incorporate two Markov decision processes: forward diffusion and the denoising process, to generate new data \cite{cao2024survey}. 
In most of the diffusion models processing tasks, such as image generation, video generation, and audio generation, the data used in the two processes is generally continuous \cite{yang2023diffusion}. However, for graph generation problems, the data belongs to a discrete space, resulting in continuous noise that cannot be directly added. 
% As a result, continuous noise cannot be directly added in the forward diffusion process. 
Therefore, in this paper, we use a discrete diffusion model to generate a sensor placement graph \cite{austin2021structured}. 

In the discrete diffusion model, the data state space $\mathcal{Z}$ of the system is discrete. For an input data $z$, which has $\vert\mathcal{Z}\vert$ potential states, i.e., $\vert\mathcal{Z}\vert$ values, we process it by one hot encoding. The one hot encoded embedding can be denoted as $\bm{z}\in \mathcal{Z}^{\vert\mathcal{Z}\vert}$.
Instead of adding random noise variables, the noises are represented by a series of transition matrices $(Q^1,Q^2,\ldots,Q^T)$, where $[Q^t]_{i,j}$ denotes 
probability of changing from state $i \in \mathcal{Z}$ to state $j \in \mathcal{Z}$. Given an input $\bm{z}^0$, the prior distribution of a sequence of increasingly noisy data points $\bm{z}^1,\bm{z}^2,\ldots,\bm{z}^T$ is
\begin{equation}
\label{eq:discrete}
\begin{split}
q(\bm{z}^{1:T}\vert \bm{z}^0)  &= \prod_{t=1}^{T} q(\bm{z}^t\vert \bm{z}^{t-1})\\
& = \prod_{t=1}^{T-1} q(\bm{z}^t\vert \bm{z}^{t-1})\cdot Q^T \\
& = \ldots\\
& = \bm{z}^0\cdot Q^1Q^2\cdots Q^T,
\end{split}
\end{equation}
where $\bm{z}^{1:T}$ represents the sequence $\bm{z}^1$, $\bm{z}^2$, $\ldots$ , $\bm{z}^T$,
$q(\bm{z}^t\vert \bm{z}^{t-1}) = \bm{z}^{t-1}\cdot Q^t$ indicates the transition from state $\bm{z}^{t-1}$ to state $\bm{z}^t$. Therefore, the distribution of noisy states $z^t$ can be calculated directly by multiplying the transition matrix $q(\bm{z}^t\vert \bm{z}^0) = \bm{z}^0\cdot \overline{Q}^t$, where $\overline{Q}^t = Q^1\cdots Q^t$. Leveraging the Bayes rule, the posterior distribution $q(\bm{z}^{t-1}\vert \bm{z}^{t}, \bm{z}^0)$ can be computed by \cite{austin2021structured}
\begin{equation}
\label{eq:posterior}
    q(\bm{z}^{t-1}\vert \bm{z}^{t}, \bm{z}^0) \propto \bm{z}^t (Q^t)' \odot \bm{z}^0~\overline{Q}^{t-1},
\end{equation}
where $\odot$ represents a pointwise product and $Q'$ is the transpose of $Q$.
Recall the goal of the original diffusion model is to transform an unknown distribution into a well-known one, such as a uniform distribution, via the forward diffusion process.
Therefore, in the discrete diffusion model, we adopt the uniform transition matrix formulated as $Q^t= \alpha^t I+(1-\alpha^t)\mathbf{1}_d \mathbf{1}'_d / d$ with $\alpha$ decreasing from $1$ to $0$. It can be proved that when 
$\lim_{t\rightarrow\infty}\alpha^t = 0$, $q(\bm{z}^t\vert \bm{z}^0)$ can coverage to a uniform distribution independently of $\bm{z}^0$ with the uniform transition $Q^t$ \cite{yang2023diffsound}.
% Considering a uniform transition formulated by $Q^t= \alpha^t I+(1-\alpha^t)\mathbf{1}_d \mathbf{1}'_d / d$ with $\alpha^t$ transitioning from $1$ to $0$, $q(\bm{z}^t\vert \bm{z}^0)$ can coverage to a uniform distribution independently of $\bm{z}^0$ when $\lim_{t\rightarrow\infty}\alpha^t = 0$
% \cite{yang2023diffsound}. 
Then, we use a neural network to learn the posterior distribution $q(\bm{z}^{t-1}\vert \bm{z}^{t}, \bm{z}^0)$, which can be utilized to recover the noisy data via the denoising process.

% given an input $z_0$ as the original data, 

\subsubsection{Graph Forward Diffusion}

% \todo{Check Symbols}

In our sensor placement optimization, we aim to generate an optimal sensor placement $\mathcal{S} = (\mathcal{V}_S, \mathcal{E}_S)$ represented by
a graph tuple $G = (V,E)$, 
where $V$ is vertices vector and $E$ represents the corresponding adjacency matrix.
Inspired by DiGress \cite{vignac2022digress}, we add noise incrementally for $T$ steps to both the vertices vector $V$ and the edge matrix $E$ via the transition matrix. 
Based on the discrete diffusion model introduced above, we denote $Q^t_V$ and $Q^t_E$ as the $t$-th transition matrices for vertices and edges, respectively. 
For any vertices and edges, the transition probabilities are defined as $Q^t_V = q(\bm{v}^t=j\vert \bm{v}^{t-1}=i)$ and $Q^t_V = q(\bm{e}^t=j\vert \bm{e}^{t-1}=i)$. 
According to Eq. \eqref{eq:discrete}, given an input graph tuple $G^0 = (V^0, E^0)$, the forward diffusion process can be formulated as
\begin{equation}
    q(G^{1:T}\vert G^0)  = \prod_{t=1}^{T} q(G^t\vert G^{t-1}) = (V^0\cdot \overline{Q}^{T}_V, E^0\cdot \overline{Q}^{T}_E),
\end{equation}
where $\overline{Q}^{t}_V$ and $\overline{Q}^{t}_E$ represent the product of the first $t$ transition matrices of vertices and edges, respectively, and $G^{1:T}$ is the graph sequence $G^1$, $G^2$, $\ldots$, $G^T$.

\begin{figure*}[htbp]
    \centering
    \includegraphics[width= 0.8\linewidth]{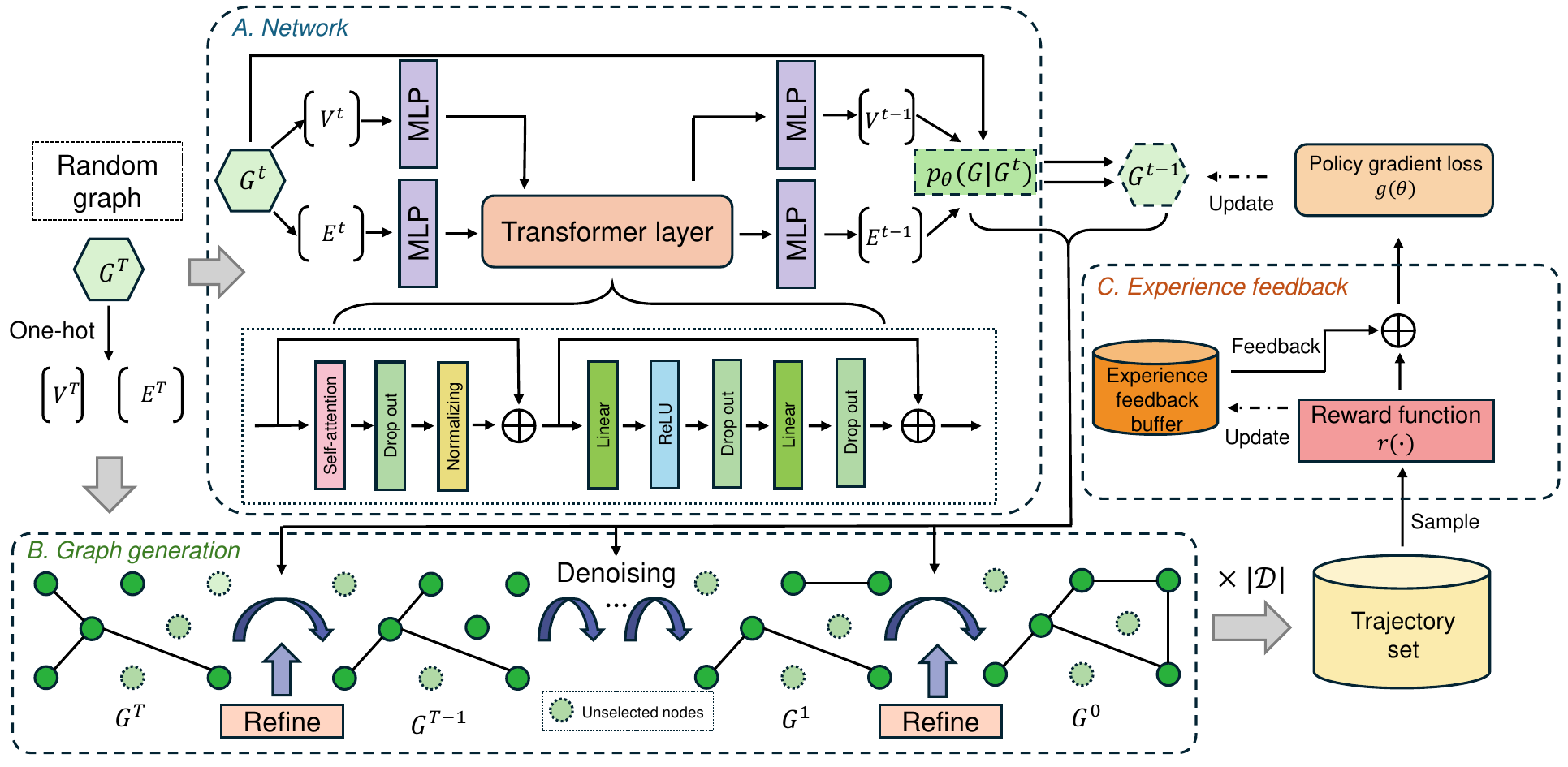}
    \caption{The illustration of EFGD framework. \textit{Part A}. The denoising network structure based on graph transformer architecture. \textit{Part B}. The graph generation process based on the denoising graph diffusion model. \textit{Part C}. The proposed experience feedback module.}
    \label{fig:framework}
\end{figure*}

\subsubsection{Graph Denoising Process}
% \todo{Change all x to v}

In general, the denoising process can be regarded as the reverse process of forward diffusion, and its goal is to accurately estimate the posterior distribution for any initial input $z^0$ in Eq. \eqref{eq:posterior}. 
Consequently, the estimated posterior distribution $p_\theta$ can be expressed as \cite{vignac2022digress}
\begin{equation}
\label{eq:sum}
\begin{split}
        &p_{\theta}(G^{t-1} | G^t) \\
        =& \prod_{i=1}^{n} p_{\theta}(v_i^{t-1} | G^t) \prod_{1 \leq i, j \leq n} p_{\theta}(e_{ij}^{t-1} | G^t).
\end{split}
\end{equation}
For each term of vertices in Eq. \eqref{eq:sum}, we can compute it by marginalization according to the initial state 
\begin{equation}
\begin{split}
        p_{\theta}(v_i^{t-1} | G^t) &= \int_{v_i} p_{\theta} (v_i^{t-1}\vert v^0_i, G^t) dp_{\theta}(v_i\vert G^t) \\
        &= \sum_{v\in \{0,1\}} p_{\theta} (v_i^{t-1}\vert v^0_i = v, G^t) \widehat{p}^V_i(v \vert G^t),
\end{split}
\end{equation}
where posterior distribution $\widehat{p}^V_i(x)$ is a prediction based on $v\in \{0,1\}$ given a noisy graph $G^t$. Referring to Eq. \eqref{eq:posterior}, we let
% \begin{equation}
% \begin{split}
%     &p_\theta (x_i^{t-1}\vert x^0_i = x, G^t) = \\
%     &
%     \begin{cases}
%         q(x^{t-1}_i \vert x^0_i = x, x^t_i), ~ &\text{if}~q(x^t_i \vert x^0_i = x)>0,\\
%         0,~&\text{otherwise.}
%     \end{cases}
%     \end{split}
% \end{equation}
\begin{equation}
p_\theta (v_i^{t-1}\vert v^0_i = v, G^t) = q(v^{t-1}_i \vert v^0_i = v, v^t_i).
\end{equation}
Similarly, for each edge, we have
\begin{equation}
\begin{split}
        p_\theta(e^{t-1}_{ij}\vert e^{t}_{ij}) &= \sum_{e\in\{0,1\}} p_\theta (e^{t-1}_{ij}\vert e^0_{ij} = e, G^t)\widehat{p}^E_{ij}(e \vert G^t),\\
        & = \sum_{e\in\{0,1\}} q(e^{t-1}_{ij}\vert e^0_{ij} = e, e^{t}_{ij})\widehat{p}^E_{ij}(e \vert G^t).
\end{split}
\end{equation}
In summary, each denoising step can be formulated as 
\cite{karras2022elucidating}
\begin{equation}
\label{eq:poster2}
    p_{\theta}(G^{t-1}\vert G^t) = \sum_{G\in \mathcal{G}} q(G^{t-1}\vert G^0 = G, G^t)\widehat{p}_{\theta}(G\vert G^t),
\end{equation}
where $\mathcal{G}$ denotes the set of all possible initial graphs, and
$\widehat{p}_{\theta}(G\vert G^t)$ is computed by a denoising network.

\subsection{Experience Feedback Graph Diffusion Policy Gradient}

\begin{figure*}[b]
\hrule
\begin{equation}
\label{eq:loss}
   g(\theta) = 
    \frac{1}{\vert \mathcal{D}\vert} \sum_{k=1}^{\vert \mathcal{D}\vert}\frac{T}{\vert \mathcal{T}_{k}\vert}\sum_{t\in\mathcal{T}_k}\left(  \frac{r(G^{0}_k)-\mu_r}{\sigma_r}H(p_{\theta}(G^{0:T}_k), p_{\theta}(G^{0}_k\vert G^{t})) 
    + \beta \cdot H(p(\widetilde{G}_k), p_{\theta}(G^{0}_k\vert G^{t}_k))\right).
\end{equation}
\end{figure*}

Based on the forward diffusion and denoising process for graph generation, we can effectively learn the posterior distribution of graph data.
However, for the sensor placement problem, we aim to obtain the posterior distribution of the optimal solution to the optimization \eqref{eq:optimization}. 
% More specifically, given an optimal graph tuple $G_o = (V_o, E_o)$, we aim to estimate the posterior distribution $q(G^{t-1}\vert G^{t}, G_o)$.
Inspired by GDPO \cite{liu2024graph}, we model the denoising process as a $T$-setp Markov decision process and solve it via RL. Given a Markov decision process $\mathcal{M} = (\mathcal{S},\mathcal{A},p_\theta,r,\rho_0)$, the proposed EFGD is defined as follows:
\begin{equation}
\label{eq:markve}
\centering
    \begin{split}
        &\bm{s}_t \triangleq (G^{T-t}, T-t), \bm{a}_t\triangleq G^{T-t-1},\\
        &\pi_\theta (\bm{a}_t\vert \bm{s}_t) \triangleq p_\theta(G^{T-t-1}\vert G^{T-t}),\\
        &r(\bm{s}_t,\bm{a}_t) \triangleq
            r(G^0),~\text{if}~t=T,\\
        &r(\bm{s}_t,\bm{a}_t)  \triangleq 0,~\text{otherwise},
    \end{split}
\end{equation}
where $\bm{s}_t \in \mathcal{S}$ and $\bm{a}_t\in \mathcal{A}$ indicate the state and action at the $t$-th step respectively, $\pi_\theta$ represents the policy for sensor placement, $p_\theta$ is the transition function determining the probabilities of state transitions, $r(\bm{s}_t,\bm{a}_t)$ denotes the reward for action $\bm{a}_t$ at state $\bm{s}_t$, $\rho_0$ gives the distribution of the initial noise graph state, and $G^0$ is the state after $T$-step denoising.

In the denoising process,
as an agent interacts in the Markov decision process, we can acquire a trajectory denoted as $\tau = (\bm{s}_0, \bm{a}_0,\bm{s}_1,\bm{a}_1,\ldots,\bm{s}_{T-1},\bm{a}_{T-1},\bm{s}_T)$. The accumulative reward $R(\tau)$ of each trajectory is given by $R(\tau) = \sum_{t=0}^{T-1}r(\bm{s}_t,\bm{a}_t) = r(G^0)$. To obtain an optimal sensor placement policy, the objective is to find a policy $\pi_\theta$ whose trajectory $\tau$ maximizes the expected accumulative reward $R(\tau)$, i.e.,
\begin{equation}
\label{eq:ACC reward}
    \mathcal{J}(\theta) = \mathbb{E}_{\tau\in p(\tau\vert \pi_\theta)}[R(\tau)] = \mathbb{E}_{G^0\in p_\theta(G^{0:T})}[r(G^0)].
\end{equation}
% \todo{add graph reward}
% \begin{prop}
% \label{prop:traj}
% Following the policy $\pi_\theta$ 
% whose trajectory $\tau$ maximizes the expected accumulative reward 
% \eqref{eq:ACC reward}, the generated sensor placement $\mathcal{S}$ 
% is the solution of the optimization problem \eqref{eq:optimization}.
% \end{prop}
% \begin{proof}
% \todo{add}
% \end{proof}
We can notice that the expected accumulative reward $R(\tau)$ in Eq. \eqref{eq:ACC reward} is equivalent to the expected reward of the final generated result. 
As a result,
we can obtain the optimal sensor placement of the proposed optimization problem \eqref{eq:optimization} via the denoising process.
% According to Proposition \ref{prop:traj},
% the dense reward function connects the value of the action taken at each step to its potential denoising result. 
% As a result, each action will aim to achieve the best denoising result, moving towards the state with the highest possible reward. 
% Finally, via denoising process, we can obtain the optimal sensor placement of the proposed optimization problem \eqref{eq:optimization}. 
To effectively generate a policy trajectory, the EFGD can utilize the framework of policy gradient methods such as the REINFORCE algorithm \cite{sutton2018reinforcement} and the proximal policy optimization (PPO) algorithm \cite{schulman2017proximal}. 
Following the policy gradient, EFGD is trained to learn the optimal policy for the placement of wireless sensors.
Given the objective $\mathcal{J}(\theta)$ , the policy gradient $\nabla_{\theta}\mathcal{J}(\theta)$ can be express as 
\begin{equation}
\label{eq:policygradinet}
\begin{split}
    \nabla_{\theta}\mathcal{J}(\theta) &= \mathbb{E}_{\tau}\left[ r(\bm{s},\bm{a}) \nabla_{\theta}\log \pi_{\theta}(\bm{a}\vert \bm{s}) \right] \\ 
    & = \mathbb{E}_{\tau} \left[r(G^{0}) \sum_{t=0}^{T-1} \nabla_{\theta}\log p_{\theta}(G^{t-1}\vert G^{t}) \right].
\end{split}
\end{equation}

However, when training based on the policy gradient in Eq. ~\eqref{eq:policygradinet}, the generated graph cannot converge to a high reward area due to the tremendous number of graph trajectories, particularly as the number of nodes increasing \cite{black2023training}. 
To solve this, we adopt the eager policy gradient in \cite{liu2024graph}, which can be written as
\begin{equation}
\label{eq:eager}
    \nabla_{\theta}\mathcal{J}(\theta) = \mathbb{E}_{\tau} \left[r(G^{0}) \sum_{t=0}^{T-1} \nabla_{\theta}\log p_{\theta}(G^{0}\vert G^{t}) \right].
\end{equation}
In the eager policy gradient, the graph trajectories are partitioned into different equivalence classes based on the possible prior graph $G^0$, where trajectories with the same $G^0$ are considered equivalent, thereby reducing the number of graph trajectories.
Even though optimizing over these equivalence classes will be much easier than the original policy gradient to explore higher reward results, it still needs a certain convergence speed due to the slow convergence nature of diffusion \cite{liu2024graph}.

To increase convergence speed, we improve the training process by modifying the policy gradient and the experience enhancement to eliminate fluctuating and unreliable policy gradient estimates.

\subsubsection{Cross-Entropy Gradient}

Inspired by the cross entropy loss, which is a measure of the difference between two probability distributions \cite{mao2023cross}, we introduce a cross-entropy gradient, which can be expressed as
\begin{equation}
\label{eq:ceg}
    \nabla_{\theta}\mathcal{J}(\theta) = -\mathbb{E}_{\tau} \left[r(G^{0}) \sum_{t=0}^{T-1} \nabla_{\theta}\big(H(p_{\theta}(G^{0:T}), p_{\theta}(G^{0}\vert G^{t}) ) \big)\right],
\end{equation}
where $H(p_{\theta}(G^{0:T}), p_{\theta}(G^{0}\vert G^{t}))$ is the cross entropy between the distribution of generated graph $p_{\theta}(G^{0:T})$ and the predicted distribution $p_{\theta}(G^{0}\vert G^{t})$. Maximizing the cross-entropy gradient is equivalent to minimizing a reward-weighted cross-entropy function, thereby exploring higher rewards and making the predicted distribution closer to the final distribution.

% The Cross entropy loss \cite{mao2023cross}

% \cite{liu2024graph}

% black2023training

% metric(flat_pred_E, flat_true_E)

% The citations in this article demonstrate a thorough understanding of the relevant literature in the field.

% Intuitively, although the number of possible graph trajectories is tremendous, if we partition them into different equivalence classes according to G0, where trajectories with the same G0 are considered equivalent, then the number of these equivalence classes will be much smaller than the number of graph trajectories. The optimization over these equivalence classes will be mucheasier than optimizing in the entire trajectory space

\subsubsection{Trajectory Experience Feedback}

Based on the RLHF, which selects the highest reward strategy as the experience feedback, we propose a trajectory experience feedback to accelerate training convergence. Different from the experience replay, the experience feedback is a generated graph with the highest reward rather than a series of action policies. Letting $p(\widetilde{G})$ as the distribution of the experience graph, the loss function can be written as
\begin{equation}
\label{eq:cem}
    l_{CEM} = \mathbb{E}_{\tau}\left[ \mathbb{E}_{t}[H(p(\widetilde{G}), p_{\theta}(G^{0}\vert G^{t}))]  \right].
\end{equation}
To minimize the loss function $l_{CEM}$, the network can use optimization solutions with higher rewards, thereby accelerating the training and convergence of the network. In practice, we select the experience graph $\widetilde{G}$ from an experience feedback buffer $\mathcal{B}$ consisting of $\vert \mathcal{B}\vert$ trajectories with the highest reward by comparing the mean square error with the current prediction $G^0$.

\new{
It is worth noting that while this mechanism shares the conceptual foundation of utilizing a replay buffer with prioritized experience replay (PER)~\cite{schaul2015prioritized}, there are fundamental distinctions in their selection criteria and objectives. 
PER prioritizes individual transitions with high temporal-difference errors to replay rare or hard-to-learn experiences. 
In contrast, our trajectory experience feedback operates at the trajectory level, prioritizing complete graph structures that yield the highest rewards. 
Instead of correcting prediction errors, these high-reward trajectories serve as explicit ``expert demonstrations'' or targets for the cross-entropy gradient, directly guiding the diffusion denoising process toward optimal solution regions.
Consequently, unlike standard PER, which is typically restricted to off-policy settings due to distribution shift issues, our proposed mechanism is compatible with on-policy algorithms such as PPO~\cite{schulman2017proximal}, as it leverages historical data as guidance targets rather than direct replay samples for policy updates.
}

Combining the proposed two training strategies together, we approximate Eqs. \eqref{eq:ceg} and \eqref{eq:cem} with a Monte Carlo estimation to obtain the loss function. Moreover, we adopt reward standardization in the
training stage to ensure the stability and convergence of network training.
The
standardized loss function can be written as Eq. \eqref{eq:loss}, where $\mathcal{D}$ and $\mathcal{T}_k$ represent the sets of sampled trajectories and timesteps, respectively, $\mu_r$ and $\sigma_r$ represent the mean value and variance of the reward for each action on the collected trajectory, respectively, and $\beta$ is an adjustable weight parameter. 

In summary, the proposed EFGD method aims to train a neural network $p_{\theta}$ to learn the posterior distribution in Eq. \eqref{eq:sum}. Via formulating the denoising process as a $T$-step Markov decision process $\mathcal{M}$, we use the policy gradient framework to 
train the network $p_\theta$ with the modified loss function as Eq. \eqref{eq:loss} to ensure the stability
and convergence of network training. Then, the parameter $\theta$ of the network $p_\theta$ is updated by stochastic gradient descent.
% \cite{friedman2002stochastic}, i.e., $\theta' = \theta + \gamma \cdot g(\theta)$, where $\gamma$ denotes the learning rate and $g(\theta)$ represents the training loss. 
Finally, we sample the Markov decision process trajectories using the trained network $p_\theta$ via denoising process in Eq. \eqref{eq:poster2}.

% thus solve the problem by gradient descent

\subsection{EFGD for Wireless Sensor Placement}
\label{sec:reward}

In this subsection, we show the data flow processing in different procedures and present our EFGD framework for wireless sensor placement. Moreover,
we outline the deployment details of EFGD for optimization, including reward function setting, network structure, and optimization process.
% Moreover, we show the data flow processing in different procedures and present our EFGD framework for wireless sensor placement.

The entire EFGD process is detailed in Algorithm \ref{alg:DRGO} and illustrated in Fig. \ref{fig:framework}. 
We aim to train a denoising network capable of estimating the posterior distribution. Each training epoch begins with the random generation of initially noised graphs. Subsequently, $\mathcal{D}$ trajectories are produced to form the trajectory set $\mathcal{D}$, originating from the initial graphs as depicted in Figure \ref{fig:framework} \textit{Part B}. During the denoising process, a refinement operation is applied to validate the generated edge matrix, ensuring its symmetry and the connectivity of only selected points. 
Graph trajectories are then sampled, and their rewards are calculated based on the reward function $r(\cdot)$. Particularly, during the reward calculation process, the refinement operation eliminates the edges that do not meet the constraint in Eq. \eqref{eq:l_constraint} to validate the generated sensor placement strategy.
Moreover, an experience feedback buffer is maintained, which stores the higher reward trajectories. We select the most suitable trajectory experience feedback based on the minimum mean square error (MMSE) between the two graphs.
This feedback, along with the trajectories, is then used to compute the policy gradient loss in Eq. \eqref{eq:loss} to train the denoising network.

% We aim to train a denoising network to estimate the Posterior distribution.
% In each training epoch, we first randomly generate initial noised graphs. Then, $\mathcal{D}$ trajectories are generated to establish the trajectory set $\mathcal{D}$ from the initial graphs based on the graph generation process as shown in Fig. \ref{fig:framework} \textit{Part B}.
% During the denoising process, we use the refine operation to validate the generated edge matrix, ensuring matrix symmetry and connecting only selected points. 
% Then, we sample graph trajrostis and calculate their reward according to the reward function $r(\cdot)$, and we We maintain the Experience feedback buffer based the reward of trajorjint. Then we select the most suitable rajectory Experience Feedback based on the MMSE of two graph. 
% Then utilize the feedback and trajectory to compute the policy gradient loss \eqref{eq:loss} to train the network.

\subsubsection{Reward Function Setting}

Recall our goal is to utilize EFGD to solve the optimization problem \eqref{eq:optimization} and obtain the optimal sensor emplacement policy. Achieving this objective requires an appropriate reward function $r$ for graph generation in the Markov decision process $\mathcal{M}$. Considering both the optimized objective and constraints in the optimization problem, we define a penalty-constrained reward function whose penalty terms are the constraints.
Specifically, the penalty-constrained reward function $r$ can be expressed as
\begin{equation}
\label{eq:rewardfunction}
\begin{split}
    &r(G^0) = \\
    &~~~~\begin{cases}
    r_1 \cdot\lambda_2(\mathcal{L}_{G^0}),~\text{if~} I_1\times I_2 = 1, \\
    -r_2 \cdot (|\mathcal{V}_{G^0}| - N) - r_3 \cdot (\lambda_s - S_{a}(\mathcal{V}_{G^0})),~\text{otherwise},
\end{cases}
\end{split}
\end{equation}
where $I_1 = \textbf{1}(|\mathcal{V}_{G^0}|\leq N)$ and $I_2=\textbf{1}(S_{a}\leq\lambda_s)$; $\textbf{1}$ denote the indicator; $r_1$, $r_2$, and $r_3$ are adjustable weight ratio. \new{Based on this formulation, an agent in the Markov decision process $\mathcal{M}$ can obtain a positive reward only if the sensor placement policy generated by the chosen action strictly satisfies the constraints. It is worth noting that $r_1$ is set to a relatively large value to act as a scaling factor. Since the Fiedler value $\lambda_2$ is typically small (ranging from 0 to 1), this scaling ensures that the gradient magnitude for objective maximization is sufficient for effective learning once the agent enters the feasible region.
Additionally, due to the different magnitudes of varying conditions, we carefully tune $r_2$ and $r_3$ to balance the penalty gradients, ensuring the agent effectively learns to satisfy all constraints.}

% Based on the penalty terms, an agent in the Markov decision process $\mathcal{M}$ can obtain a non-zero reward only if the sensor placement policy generated by the chosen action satisfies the constraints. 
% Additionally, due to the different magnitudes of different conditions, we use adjustable weight ratios $r_i$ to ensure learning capabilities for different constraints.

% Additionally, since the optimized tree height $h$ is between $1$ and $N-1$, we use the difference between $N$ and the tree height $h$ as the reward. Thus the smaller the tree height, the greater the reward.
% Additionally, since the height $h$ of the tree is between $1$ and $N-1$, we use the difference between $N$ and the tree height $h$ as the reward. Thus the smaller the tree height, the greater the reward.
Consequently, for a Markov decision process $\mathcal{M}$ and the penalty-constrained reward function $r$, the denoising result following the trajectory of the maximum cumulative reward will correspond to the optimal solution of the optimization problem \eqref{eq:optimization}, i.e., the optimal robust sensor placement policy.

\subsubsection{Network Structure}

We construct the denoising network for EFGD via the graph transformer architecture based on \cite{liu2024graph} as shown in Fig. \ref{fig:framework} \textit{Part A}.
Firstly, one-hot embeddings of nodes and edges are extracted from the graph $G^t$, respectively, to convert state information into probability space to facilitate prediction of node and edge selection. These embeddings are then mapped to latent space through the Multi-layer Perception (MLP) layer, which consists of a combination of two consecutive linear layers and Rectified Linear Unit (ReLU) activation functions. Subsequently, the processed node and edge data are input into a transformer structure together, using a self-attention mechanism based on a Feature-wise Linear Modulation layer \cite{liu2024graph} to focus on each detail feature. Simultaneously, the dropout layer is utilized to prevent the network from being overfitted. Residual connections and normalizing layers are combined to facilitate network learning, enhance stability, and prevent gradient explosion and disappearance.
Finally, the MLP layer is used to restore the original probability space $p_{\theta}(G|G^t)$ to predict the current selection of the edges and nodes $G^{t-1}$.

\begin{algorithm}[tpb]
    {\small \caption{\textcolor{black}{The proposed EFGD algorithm}} \label{alg:DRGO}
    $\mathbf{Input}$: Initial denoising network $p_\theta$, a Markov decision process $\mathcal{M} = (\mathcal{S}, \mathcal{A}, p_\theta, r, \rho_0)$, \# of denoising steps $T$, \# of trajectory sample $\vert\mathcal{D}\vert$, \# of timestep samples $\vert \mathcal{T}\vert$, learning rate $\gamma$, \# of training steps $N_t$, \# of the size of experience feedback buffer $|\mathcal{B}|$;\\
    \vspace{0.5em}
    \textit{Procedure 1: EFGD Training};\\
    \quad  \For{$i = 1, 2, \dots, N_t$}{    
    \quad    \For{$d = 1, 2, \dots, |\mathcal{D}|$}{
    \quad     Sample initial noisy graph from the Markov decision process $\mathcal{M}$\\
    \quad \For{$t = 1, 2, \dots, T$}{
    \quad                Perform denoising based on Eq. \eqref{eq:poster2}
    } 
    \quad Acquire trajectory $\boldsymbol{\tau}_d$ \\
    \quad Sample timestep $T_d$ $\sim$ $\mathrm{Uniform}([1, T])$\\
    \quad Calculate reward $r(G^0)$ based on Eq. \eqref{eq:rewardfunction}
    \quad Update experience feedback buffer $\mathcal{B}$ based on the reward $r(G^0)$
    }
    }
    \quad Calculate policy gradient loss $g(\theta)$ based on Eq. \eqref{eq:loss}\\
    \quad Update $p_\theta$ parameters by gradient descent\\
    \vspace{0.5em}
    \textit{Procedure 2: EFGD Inference};\\
    \quad \For{$t = 1, 2, \dots, T$}{
    \quad Perform denoising process based on Eq. \eqref{eq:poster2}\\
    }
    $\mathbf{Output}$: Generated sensor placement $G^0$
}\end{algorithm}

\section{Numerical Results}
\label{sec:exp}

In this section, we implement the proposed sensor placement problem in a real power system, IEEE 118-Bus System \cite{pypower}. Then, we conduct extensive experiments to evaluate the proposed EFGD algorithm,
covering two aspects: algorithm performance and the impact of hyperparameters on performance.

\subsection{Experiment Setting}
\label{sec:exp_set}

The experiments are conducted on a Linux server with an NVIDIA A100 GPU with 80 GB of memory. 
All codes are written in Python, network training is based on the PyTorch package, and power grid simulation uses the pypower package \cite{pypower}.
Due to the sensor placement problem, we cannot get real-time monitoring data from sensors. Therefore, we generate a dataset of anomalies and normal scenarios in the IEEE 118-Bus System through the pypower simulator with the same power load setting in \cite{hooi2019gridwatch}. 
By considering resistance information as distance information, we utilize Multidimensional Scaling \cite{kruskal1978multidimensional} to simulate the 2D coordinates of the nodes in the power grid. Then, we can calculate the distance between any two nodes via 2D coordinates.
For the LNSPL model, we utilize the parameters of a smart grid substation model with IEEE 802.15.4 \cite{sandoval2017improving}.
Specifically, the reference distance $d_0$ is $1$m, the reference path loss is set at $40.3308$ dB, the path loss exponent $\gamma$ is 1.701, and the variance $\sigma$ of the Gaussian noise $X_{\sigma}$ is set at $2.18$ dB.

\new{
Furthermore, to determine the reliability threshold $\lambda_c$ for practical applications, we perform a link budget analysis based on hardware specifications. 
We set the transmission power to 10 dBm and the noise power to -90 dBm. To ensure a lower transmission error, the SNR must exceed a minimum requirement of 25 dB \cite{gungor2010opportunities}. 
Based on Eq. \eqref{eq:snr}, we calculate that links with a path loss greater than 75 dB are unstable. 
This calculation provides a guideline for setting the threshold $\lambda_c$ in Eq. \eqref{eq:l_constraint}, which is set to 75 dB in our experiment.
For anomaly detection, the parameter settings are guided by deployment budgets and safety standards. 
We assume a budget cap of 25 sensors, i.e., $N=25$ in Eq. \eqref{eq:O_C_3}. 
The threshold $\lambda_a$ for the anomaly detection score $S_{a}$ is set at 50~\cite{hooi2019gridwatch}, which is determined by the historical statistical characteristics of the grid.
Finally, the safety threshold $\lambda_s$ reflects the operational requirement for detection performance. 
Targeting a detection accuracy of 90\% to satisfy the grid safety standard, we set $\lambda_s=0.90$ in Eq.~\eqref{eq:T_constraint}.
}
% Furthermore, we set that the transmission power is 10 dBm, the noise power is -90 dBm, and the signal-to-noise ratio in Eq. \eqref{eq:snr} must exceed 25 dB to ensure a lower transmission error \cite{gungor2010opportunities}. Therefore, 
% the links whose path loss between nodes is greater than 75 dB are unstable and may fail,
% which establishes the threshold \(\lambda_c\) in Eq. \eqref{eq:l_constraint} as 75 dB. 
% For anomaly detection, we assume that the budget cap for the number of sensors is 25, i.e., \(N=25\) in Eq. \eqref{eq:O_C_3}. The threshold \(\lambda_a\) for the anomaly detection score \(S_{a}\) is set at 50. 
% Additionally, the accuracy of the detection is 90\%, which sets \(\lambda_s=0.90\) in Eq. \eqref{eq:T_constraint}.

For training, we set all methods' batch size and learning rate to $256$ and $1\times10^{-5}$, respectively. Then, we train each learning method for $90$ epochs, with the reward function and network structure discussed in Subsection \ref{sec:reward}. The weight ratios are set as $r_1= 5000$, $r_2=1.075$, and $r_3=0.5$.
\new{
Crucially, during training, the path loss is randomly generated using the LNSPL model at each iteration. This Monte Carlo-based simulation approach ensures that the learned policy captures the channel's stochastic nature and is robust to shadowing variations~\cite{harrison2010introduction}.
In the inference stage, to comprehensively assess performance, we generate $50$ distinct sensor placement strategies using the trained model. We then evaluate the effectiveness of these strategies by calculating the average reward and variance under $100$ independent test conditions with random shadowing realizations.
}

% Furthermore, we suppose the communication links whose path loss is bigger than $52.7386$ dB, (less than the 1/4 diameter of the graph), are Not stable, may fail, i.e., the threshold $\lambda_c$ in Eq. \eqref{eq:l_constraint} is $52.7386$ dB. 
% As for anomaly detection, we suppose the budget cap of the number of sensors is $40$, i.e., $N=40$ in Eq. \eqref{eq:O_C_3}. 
% The threshold $\lambda_a$ for anomaly detection score $S_{a}$ is set at 50.
% The accuarauy of detection is $90\%$, i.e., $\lambda_s=0.9$ in Eq. \eqref{eq:T_constraint}.
% Furthermore, 
% we assume that distances between nodes greater than one-quarter of the graph's diameter are unstable and may fail.

\subsection{Performance Comparison}

\subsubsection{Baseline Methods}

We benchmark the proposed EFGD algorithm with the following baselines:
\textbf{Greedy-Accuracy}, where the greedy strategy based on GridWatch is adopted to ensure the accuracy of anomaly detection \cite{hooi2019gridwatch};
\textbf{Greedy-Robustness}, where the greedy algorithm identifies the optimal location and link connection for the new node by iteration, aiming to maximize the Fiedler value of sensor placement;
\textbf{Random}, 
where the sensor placement nodes and network connections are randomly selected;
\textbf{GDPO}, which follows the graph diffusion optimization framework in \cite{liu2024graph}; 
\textbf{DDPO}, which is a deep RL algorithm based on standard negative log policy gradient \cite{black2023training}. 
\new{
Note that the RL baselines, including GDPO and DDPO, are set to the same setting as the proposed method EFGD and are trained using the identical reward function in Eq.~\eqref{eq:rewardfunction} to ensure a strictly fair comparison.
Additionally, for traditional methods (Greedy and Random), we also utilize Eq.~\eqref{eq:rewardfunction} as a unified evaluation metric to calculate their ``AvgReward.'' 
This standardized quantification allows for a direct assessment of how well each method satisfies the constraints and maximizes the robustness objective under the same criteria.
}
% Note that the RL baselines, including GDPO and DDPO, are set to the same setting as the proposed method EFGD. 
Additionally, for the proposed method EFGD, the diffusion step is 20, the buffer length is 50, and the weight parameter $\beta=0.2$.

\subsubsection{Sensor Placement Optimization}

Fig. \ref{fig:ex1} illustrates the average reward (denoted as AvgReward) achieved by the EFGD and compared to the baseline learning methods DDPO and GDPO, where EFGD, DDPO, and GDPO reward curves are represented in red, orange, and yellow, respectively.
Initially, when the condition $I_1 \times I_2 = 1$ in the reward function (Eq. \eqref{eq:rewardfunction}) is mostly unsatisfied, indicating the phase where the models are learning the constraints, the performances of all three methods are comparable, with GDPO slightly outperforming the others. As the training progresses beyond the constraint learning phase and into the objective function learning phase, where the models start generating graphs that satisfy the constraints around 40 epochs, there is a significant reward improvement in all models compared to their initial performances with the increment in positive samples.
Moreover, the average reward for the EFGD strategy shows a remarkable increase compared to the other two benchmarks.
Notably, EFGD converges after 60 epochs, while DDPO and GDPO take approximately 74 epochs.
This observation underscores that additional experience feedback can effectively enhance the model's convergence speed by 18.9\%.
Comparing the average AvgReward since the convergence epoch, EFGD reaches an average AvgReward of 2.2007, while DDPO and GDPO converge with an average reward of -9.9470 and -8.4694, respectively.
With -63 as the baseline average starting reward, the proposed EFGD algorithm can enhance the average reward by 22.90\% compared to DDPO, and by 19.57\% compared to GDPO.

\begin{figure}[htbp]
    \centering
    \includegraphics[width= 0.8\linewidth]{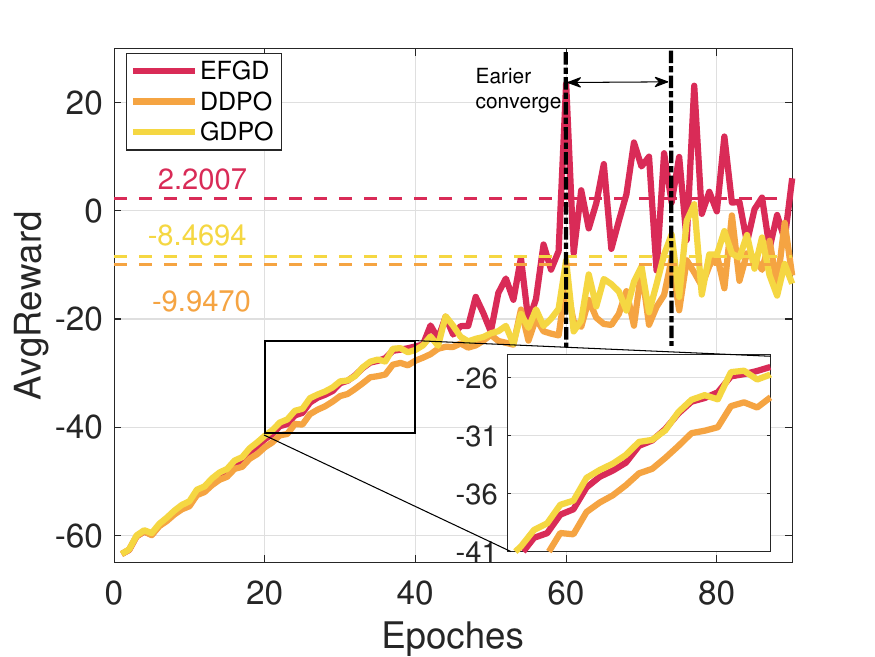}
    \caption{
     The average reward (AvgReward) of EFGD and other baselines over training epochs.}
    \label{fig:ex1}
\end{figure}

\begin{figure}[htbp]
    \centering
    \includegraphics[width= 0.8\linewidth]{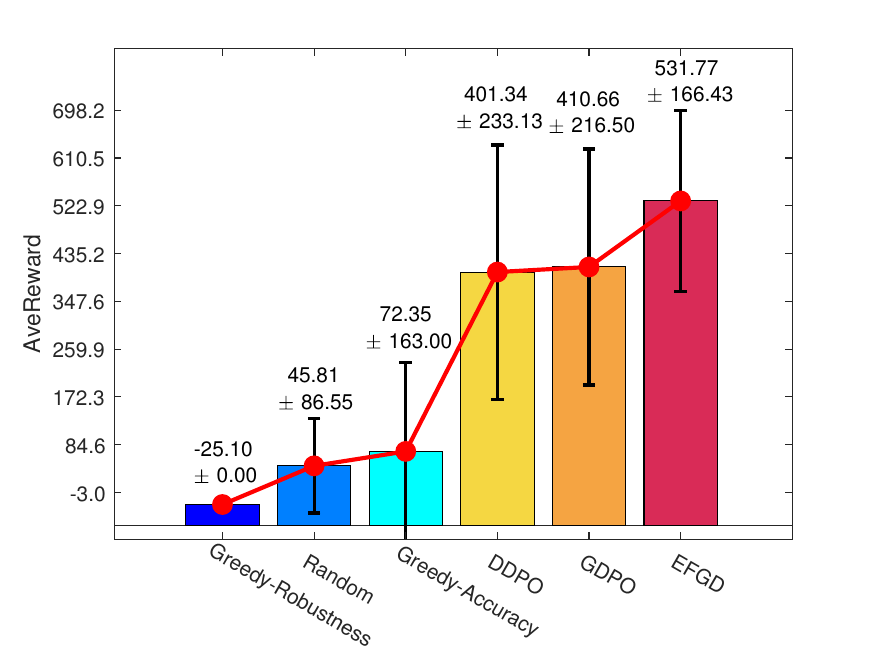}
    \caption{
     The average reward (AvgReward) of EFGD and other baselines in the inference stage.}
    \label{fig:ex12}
\end{figure}

Further analysis is conducted using these aforementioned algorithms to generate 50 sensor placement strategies, evaluating the average reward and variance of these strategies under 100 test conditions. 
The results for the optimal placement strategy for each method are plotted in Fig. \ref{fig:ex12}. The traditional baseline methods show a stark contrast in performance: Greedy-Robustness scores the lowest at $-25.10 \pm 0.00$, Greedy-Accuracy has a performance of $72.35 \pm 163.00$, and Random strategy yields $45.81 \pm 86.55$. 
In comparison, the RL-based methods significantly outperformed these, with DDPO achieving $401.34 \pm 233.13$, GDPO $410.66 \pm 216.50$, and EFGD excelling at $531.77 \pm 166.43$.
The results highlight that the proposed EFGD surpasses the learning baselines GDPO and DDPO by 25.57\% and 28.09\%, respectively, and outperforms traditional methods by over 400\%. 
The primary reason why the performance of the two greedy algorithms is similar to that of the random algorithm yet significantly lower than that of the learning-based method is that each greedy algorithm focuses solely on one aspect of the system, such as accuracy or robustness. This narrow focus makes it challenging to meet the overall constraints of the optimization or to fully maximize the optimization goal, thereby demonstrating that traditional methods are inadequate for such complex NP-hard challenges. 
Furthermore, the substantial improvement over other learning underscores the efficacy of incorporating high-reward strategies as experience feedback into the model, which enables a sharper focus on such strategies and facilitates convergence to more optimal solutions. Moreover, the reduced standard deviation of the reward from the generated placement policy indicates that the sensor placement produced by EFGD exhibits more robust performance across varying conditions.

% With the experiment feedback, This capability is crucial for addressing the resilient sensor placement problem, which is categorized as NP-hard, thereby demonstrating that traditional methods are inadequate for such complex challenges.

% traditional methods manage only a baseline improvement of more than 400\% when compared to the initial condition, whereas learning-based approaches demonstrate superior enhancements: EFGD by 25.57\%, GDPO by 28.09\%, and DDPO by an improvement reflective of learning-based enhancements.

% The substantial improvement over traditional methods underscores the efficacy of incorporating high-return strategies as feedback into the model, which enables a sharper focus on such strategies and facilitates convergence to more optimal solutions. This capability is crucial for addressing the resilient sensor placement problem, which is categorized as NP-hard, thereby demonstrating that traditional methods are inadequate for such complex challenges.

Fig. \ref{fig:inference} illustrates the graph generation process of the trained EFGD model under varying communication conditions. The displayed results confirm that EFGD adeptly adjusts its node and edge generation processes during the denoising process. This capability enables the generation of graphs that effectively meet the requirements imposed by diverse scenarios, thereby demonstrating the model's extensive generalizability.
Moreover, it is noteworthy that in the generated graphs, nodes are interconnected by multiple edges. This structural characteristic significantly enhances the robustness of the network, ensuring that the graph remains connected even in the event of edge failures. Such robustness is critical in maintaining the integrity and operational stability of the system under adverse conditions.

\new{
\subsubsection{Computational Efficiency and Practical Applicability}

Beyond the theoretical convergence property, we further evaluate the computational efficiency in terms of wall-clock time to assess real-world applicability. 
For the training phase, the average computational time is approximately 60 seconds per epoch. 
Since the proposed EFGD algorithm typically converges within 60 epochs, as shown in Fig.~\ref{fig:ex1}, the entire model training process can be completed in approximately one hour. 
Considering that sensor placement is typically a long-term planning task performed offline, this training overhead is highly acceptable for practical engineering deployments.
For the inference phase, generating a single robust sensor placement strategy takes only about 1 second. 
This ultra-low latency enables system operators to perform rapid ``what-if'' analyses or dynamically regenerate strategies in response to significant topology changes, demonstrating the strong potential of EFGD for efficient practical applications in CPPS.
}
% Fig. \ref{fig:inference} displays the generated
% results under different communication conditions.
% The generated graphs indicate that the trained EFGD can
% adapt the node and edge generations 
% according to the inputs,
% hence generating the desired graphs under various conditions,
% demonstrating the generalizability of the proposed EFGD
% At the same time, we can find that the nodes of the generated graph are connected by more than one edge, which ensures the ability of the graph to remain connected when edge failures occur, ensuring the resilience of the system.

\begin{figure*}[htbp]
    \centering
    \includegraphics[width= 0.95\linewidth]{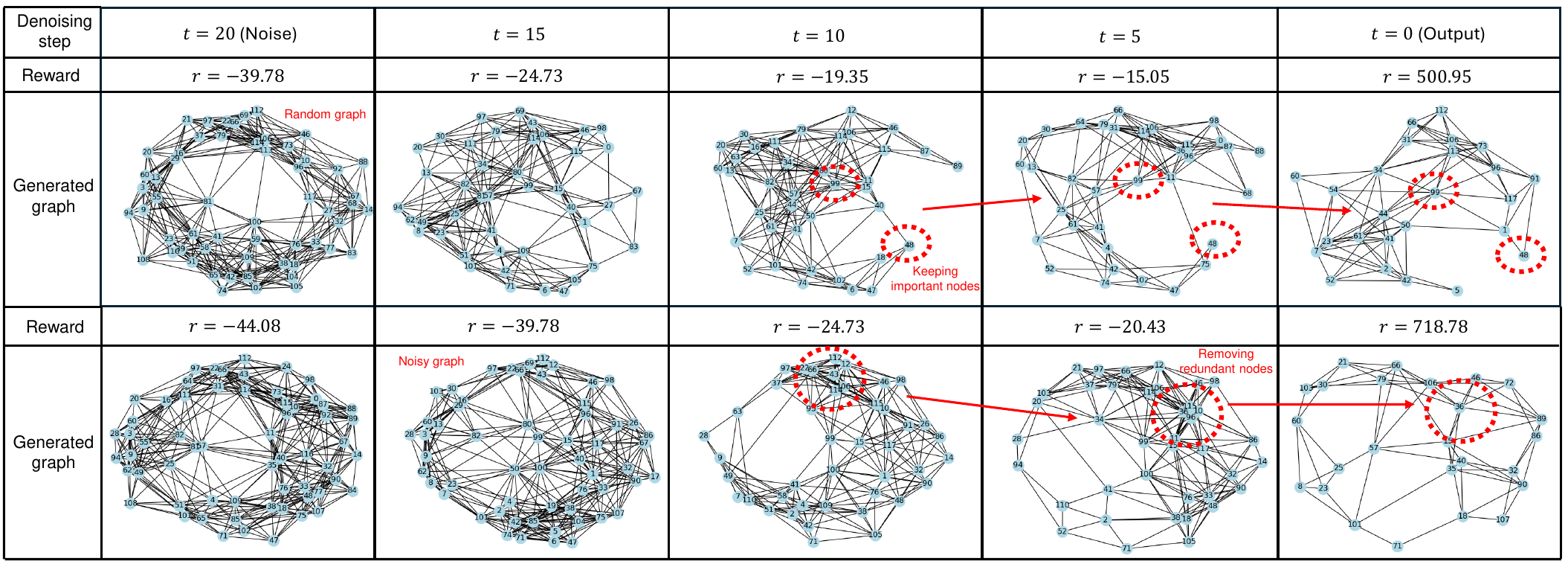}
    \caption{The illustration of the graph generation process in IEEE-118 bus system under two communication conditions.}
    \label{fig:inference}
\end{figure*}

\subsection{Hyperparameters Analysis}

Next, we analyze the impact of model hyperparameters on the performance of the proposed model from three aspects, including 1). the number of diffusion steps $T$; 2). the value of weight parameter $\beta$; and 3). the size of experience feedback buffer $\vert\mathcal{B}\vert$.

\subsubsection{Diffusion Step}

First, our analysis focuses on the impact of the diffusion step size on the reward achieved by the model's generation strategy. The diffusion step size generally influences the generation quality of diffusion models. A longer step size typically pushes the noise-added distribution closer to the standard normal distribution, facilitating the learning of an accurate posterior distribution.
As depicted in Fig. \ref{fig:ex2}, the diffusion step labeled as ``EFGD-T" varies, where ``T" represents the number of diffusion steps. With the exception of ``EFGD-5", which shows a slower convergence, the trends of the other three curves (``EFGD-10", ``EFGD-20", ``EFGD-30") are remarkably similar, all converging within 60 epochs. 
Furthermore, 
The average AvgRewards post-convergence are as follows: EFGD-5: -6.9043; EFGD-10: 2.0901; EFGD-20: 2.2007; EFGD-30: -0.4780.
These results indicate that there is generally an upward trend in the reward as the noise addition step length increases. However, performance declines when the diffusion step length becomes excessively large. This decline can be attributed to the fact that the final denoising result, $G^0$, is used as the trajectory's reward and the basis for calculating the policy gradient. While this approach emphasizes the overall scenario, it overlooks the finer details, particularly the reward attributed to each step. As the diffusion step length increases, the exploration capabilities of the model are somewhat diminished, thereby adversely affecting performance. In summary, we need to select an appropriate de-noising step size to balance details and global features.

% First, we analyze the impact of the diffusion step size on the reward of the model generation strategy. Usually, the generation quality of the diffusion model is directly related to the denoising step size. This is because a longer step size can make the distribution after adding noise closer to the standard normal distribution, thereby learning an accurate posterior distribution.
% As shown in Fig. \ref{fig:ex2}, 
% for the $T$ diffusion step is denoted as ``EFGD-T", 
% Except for diffusion-5 converging slowly, the trends of the other three curves are basically the same and converge in 60 epochs. By calculating the average AvgReward after convergence,

% EFGD-5 -10.8933, EFGD-10 1.1275 EFGD-20 2.9646, EFGD-30 -0.4379.

% The results show that as the noise addition step length increases, the reward obtained by the model is basically on an increasing trend. However, as the denoising step length is too large, the performance tends to decline. This is partly because we use the final denoising result $G^0$ as the reward of the trajectory and calculate the policy gradient. While paying more attention to the overall situation, we ignore the details, that is, the reward obtained for each step. As the denoising step length increases, the exploration ability of the model decreases to a certain extent, resulting in a decline in performance.

% \begin{figure}[htbp]
%     \centering
%     \includegraphics[width= 0.8\linewidth]{figs/ex3.pdf}
%     \caption{
%    The average reward (AvgReward) of EFGD under different diffusion steps.}
%     \label{fig:ex2}
% \end{figure}

\begin{figure}[htbp]
\centering
\begin{subfigure}{.32\textwidth}
  \centering
  \includegraphics[width=1.0\linewidth]{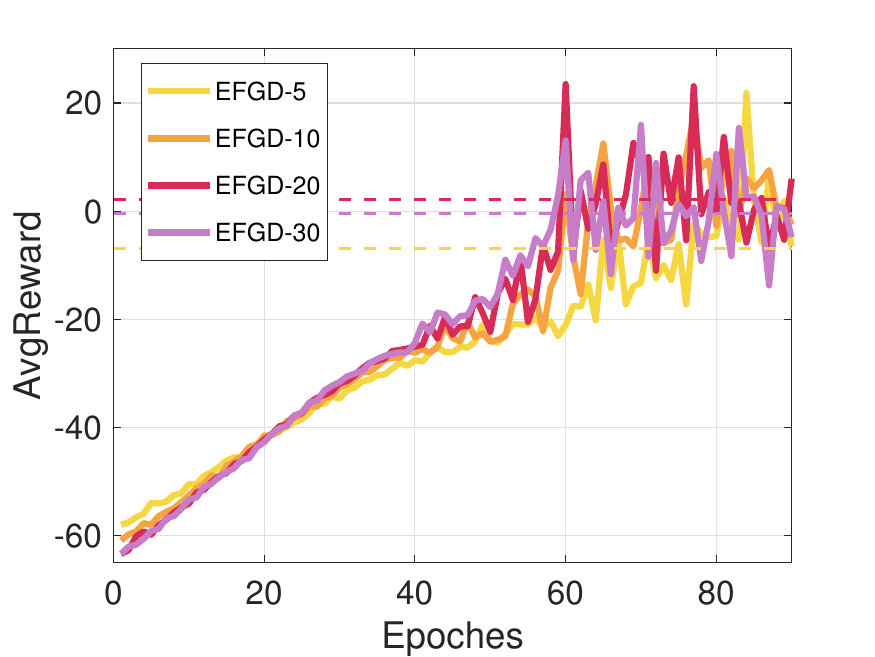} 
  %\caption{BER vs $\alpha$ with SNR = 5 dB}
  %\label{fig:sub1}
\end{subfigure}%
\begin{subfigure}{.16\textwidth}
  \centering
  \includegraphics[width=1.0\linewidth]{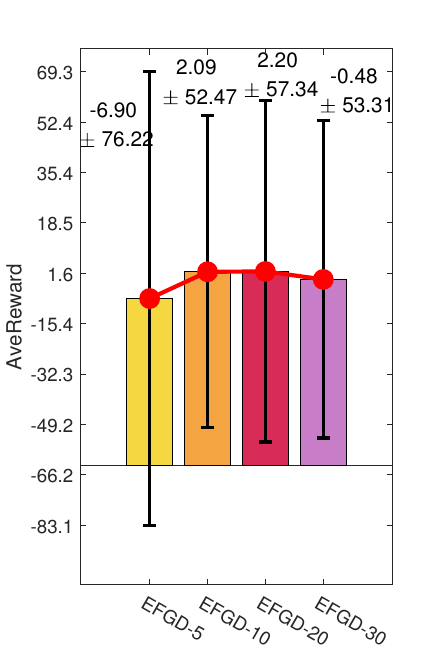}
  %\caption{BER vs $\alpha$ with SNR = 25 dB}
  %\label{fig:sub2}
\end{subfigure}
\caption{The average reward (AvgReward) of EFGD under different diffusion steps.}
\label{fig:ex2}
\end{figure}

\subsubsection{Weight Parameter}

% Then, 
% We analyzed the impact of weight parameter $\beta$ on model performance. The weight parameter affects the impact of experience feedback on the current policy gradient. As it increases, the weight of feedback also increases. as shown in Fig. \ref{fig:ex3}, when the weight is 0.1, the curve "EFGD-0.1" converges to 70 epochs, which is higher than the original 74 epochs of GDPO and DDPO, and as the weight increases, the convergence epoch is basically 60. By calculating the average AvgReward after convergence, we can find that as the weight increases, the average value is decreasing. Specifically, EFGD-0.1 is 3.9522, EFGD-0.2 is 2.2007, -0.3 is 2.4789, and -0.4 is 1.8746. The performance of EFGD-0.4 is 3\% lower than that of EFGD-0.1

% This shows that as the feedback weight increases, the model will pay more attention to the strategies that have been explored, thereby reducing the model's ability to explore strategies and making it easier for the model to converge to a local optimal solution, resulting in a decrease in model performance.

We then investigate the effect of the weight parameter $\beta$ on model performance, which modulates the influence of experience feedback on the current policy gradient. As this parameter increases, so does the impact of the feedback on the gradient calculations.
Fig. \ref{fig:ex3} presents the convergence behavior of the model at various weights. At a 
$\beta$ value of 0.1, the model ``EFGD-0.1" converges at 70 epochs, demonstrating a faster convergence than the original 74 epochs observed for GDPO and DDPO. As the weight parameter increases further, the convergence consistently occurs around 60 epochs.
Upon examining the average AvgReward after convergence, a trend emerges where the average reward decreases as the weight parameter increases. Specifically, The average rewards for EFGD with parameters -0.1, -0.2, -0.3, and -0.4 are as follows: 3.9522, 2.2007, 2.4789, and 1.8746, respectively.
The performance of ``EFGD-0.4" is 3\% lower than that of ``EFGD-0.1". This decrement in performance suggests that with higher feedback weights, the model increasingly focuses on previously explored strategies. This shift in focus can reduce the model's capability to explore new strategies, potentially leading to premature convergence to local optima and, thus, a decline in overall performance.

% \begin{figure}[htbp]
%     \centering
%     \includegraphics[width= 0.8\linewidth]{figs/ex2.pdf}
%     \caption{
%     The average reward (AvgReward) of EFGD under different weight parameters.}
%     \label{fig:ex3}
% \end{figure}

\begin{figure}[htbp]
\centering
\begin{subfigure}{.32\textwidth}
  \centering
  \includegraphics[width=1.0\linewidth]{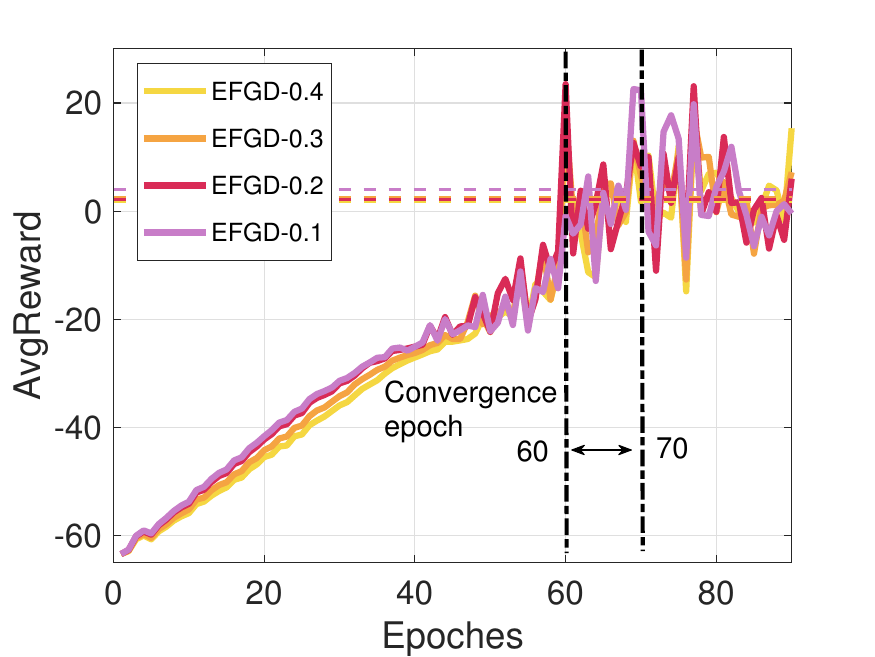} 
  %\caption{BER vs $\alpha$ with SNR = 5 dB}
  %\label{fig:sub1}
\end{subfigure}%
\begin{subfigure}{.16\textwidth}
  \centering
  \includegraphics[width=1.1\linewidth]{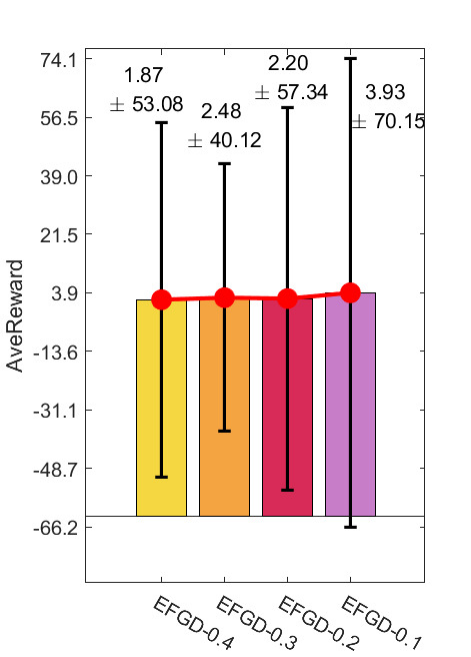}
  %\caption{BER vs $\alpha$ with SNR = 25 dB}
  %\label{fig:sub2}
\end{subfigure}
\caption{The average reward (AvgReward) of EFGD under different weight parameters.}
\label{fig:ex3}
\end{figure}

\subsubsection{Size of Experience Feedback Buffer}

% Finally, we evaluate the impact of the Size of Experience Feedback Buffer on model performance. As the size of Experience Feedback Buffer increases, Experience Feedback Buffer can store more good strategies that have appeared previously, and use MMSE as a similarity evaluation criterion to assist model learning when calculating policy gradient. As shown in Fig. \ref{fig.ex4}, the trends of all curves are basically the same and all converge in 60 epochs. Calculating the average AvgReward after convergence, within a certain limit, as the buffer size increases, the performance of the model can be improved. Specifically, when the buffer size is 10, 20, and 50, it is 1.7152, 1.8329, and 2.2007 respectively. However, when the buffer size is too large, such as 1000, the average AvgReward is 0.1120, which has a significant decrease.

% As the buffer size increases, the model can find more appropriate and reasonable experience feedback based on MMSE, thereby accelerating convergence and improving model performance. However, as the buffer size becomes too large, many local optimal solutions enter the buffer and cannot be replaced with trajectory sampling, causing many trajectories to use them as experience, making it easier for the model to fall into local optimality. Therefore, choosing an appropriate buffer size requires paying attention to the richness of optimal solutions and the trade-off of local optimal solutions.

Lastly, we evaluate the role of the experience feedback buffer size in shaping model performance. An increase in buffer size allows the storage of a larger number of previously successful strategies, which the model utilizes as references when calculating the policy gradient using the MMSE as a similarity criterion.
Fig. \ref{fig:ex4} illustrates that regardless of buffer size, all variations of the model converge around 60 epochs. Furthermore, the average AvgReward after convergence demonstrates that within a certain range, an increase in buffer size enhances model performance. Specifically, with buffer sizes of 10, 20, and 50, the AvgRewards are 1.7152, 1.8329, and 2.2007, respectively. Conversely, an excessively large buffer size, such as 1000, results in a significant reduction in performance, with an average AvgReward of 0.1120.
This decrease in performance with larger buffer sizes can be attributed to the buffer's capacity to accumulate many local optimal strategies that are not effectively replaced through trajectory sampling. As a result, many trajectories may rely excessively on these suboptimal strategies as feedback, thereby predisposing the model towards local optima. Thus, while an increased buffer size can facilitate faster convergence and initially improve performance by providing richer experiential feedback, there is a critical balance to be maintained. It is vital to optimize the buffer size to ensure a diversity of optimal solutions and mitigate the dominance of local optima, thus maintaining the efficacy and generalizability of the model.

% \begin{figure}[htbp]
%     \centering
%     \includegraphics[width= 0.8\linewidth]{figs/ex4.pdf}
%     \caption{
%     The average reward (AvgReward) of EFGD under different sizes of experience feedback buffer.}
%     \label{fig:ex4}
% \end{figure}

\begin{figure}[htbp]
\centering
\begin{subfigure}{.32\textwidth}
  \centering
  \includegraphics[width=1.0\linewidth]{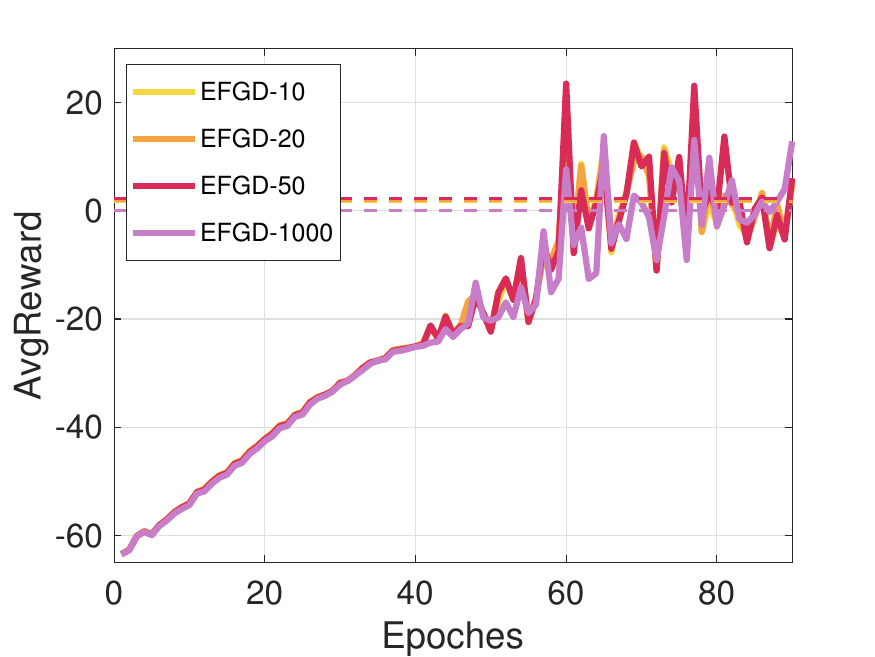} 
  %\caption{BER vs $\alpha$ with SNR = 5 dB}
  %\label{fig:sub1}
\end{subfigure}%
\begin{subfigure}{.16\textwidth}
  \centering
  \includegraphics[width=1.0\linewidth]{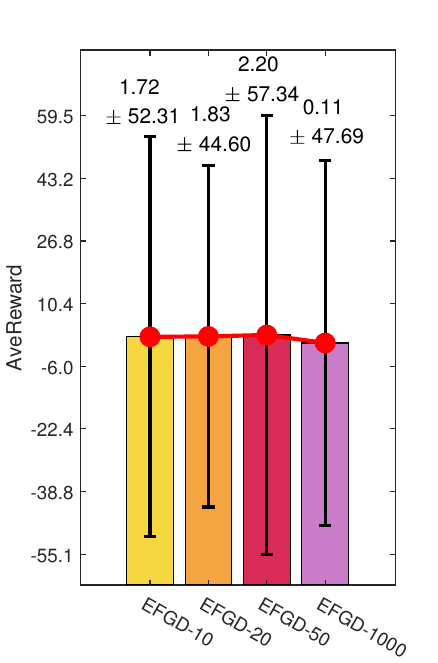}
  %\caption{BER vs $\alpha$ with SNR = 25 dB}
  %\label{fig:sub2}
\end{subfigure}
\caption{The average reward (AvgReward) of EFGD under different sizes of experience feedback buffer.}
\label{fig:ex4}
\end{figure}
\section{Conclusion}
\label{sec:con}

% In this paper, we have presented a systematic method for the
% efficient provision of zero-trust services in NGN. Specifically,
% we have modeled zero-trust networks via hierarchical graphs,
% leveraging micro-segmentations for network organization and
% SFCs for service executions. Based on this framework, we have
% presented the LEGD algorithm, which leverages graph diffusion,
% policy optimization, and LLM enhancement to realize the
% utility-controlled micro-segmentation generation. Furthermore,
% we have proposed LEGD-AM, providing an adaptive way
% to perform task-oriented fine-tuning on LEGD to adapt to
% new environments/requirements. Extensive experimental results
% confirmed the effectiveness of our proposals.

% In this paper, we have presented a resilient sensor placement optimization in CPPS. 
% Specifically, we have modeled the sensor placement problem via graph, leveraging LNSPL model to ensure the reliably of data transmission, and the Fiedler value to measure the resilience of the graph under the line failure, and three anomaly detectors to ensure the safety of system.

% Based on this framework,
% we have proposed the EFGD algorithm, which utilizes a graph diffusion model combing corss-entropy gradient and enpperience feedback to solve the proposed
% resilient sensor placement optimization problem that maximizes network robustness while ensuring accurate anomaly detection.

% Simulation results show that EFGD outperforms traditional methods in optimizing sensor placements, thus ensuring robust and reliable CPPS operations.

In this work, we have presented a novel approach for optimizing robust sensor placement within CPPS. Our approach involved modeling the sensor placement challenge as a graph-based optimization problem, utilizing the LNSPL model to ensure reliable data transmission, the Fiedler value to assess graph robustness against line failures, and employing three anomaly detectors to enhance system safety. We first proved the proposed optimization problem is NP-hard.
To address this complex optimization, we have proposed the EFGD algorithm, which employs a graph diffusion model integrated with cross-entropy gradient and experience feedback mechanisms.
By combining the experience feedback, the proposed EFGD can converge faster to find a better solution with higher rewards tailored to effectively solve the robust sensor placement optimization problem. Several simulation results show that EFGD outperforms traditional methods in optimizing sensor placements, thus ensuring robust and reliable CPPS operations.

\bibliography{Ref}

\end{document}